\documentclass[sigconf,10pt]{acmart}

\usepackage{balance}

\def\BibTeX{{\rm B\kern-.05em{\sc i\kern-.025em b}\kern-.08emT\kern-.1667em\lower.7ex\hbox{E}\kern-.125emX}}

\usepackage{microtype}
\usepackage{graphicx}
\usepackage{color}
\usepackage{url}
\usepackage{hyperref}
\usepackage{xspace}
\usepackage{enumitem}
\usepackage{epstopdf}
\usepackage{appendix}
\usepackage{amsmath}
\usepackage{empheq}
\usepackage[mathcal]{euscript}
\usepackage{bbm}  
\usepackage{bbold}  
\usepackage{dsfont}
\usepackage{multirow}
\usepackage{tabularx}
\usepackage{booktabs}
\usepackage{relsize}
\usepackage[labelfont=bf,textfont=bf,justification=raggedright]{subcaption}
\usepackage[labelfont={bf},textfont={bf}]{caption}
\usepackage[T1]{fontenc}
\usepackage{turnstile}
\usepackage{soul}
\usepackage{wrapfig}
\usepackage{float}
\usepackage[normalem]{ulem}
\usepackage{fancyvrb}
\usepackage{mdframed}
\usepackage{framed}
\usepackage{centernot}
\usepackage{mathtools}
\usepackage{setspace}
\usepackage{makecell}

\usepackage[all]{nowidow}

\usepackage{marginnote}
\usepackage{geometry}
\makeatletter
\let\oldmarginnote\marginnote
\renewcommand*{\marginnote}[1]{
   \begingroup
   \ifodd\value{page}
     \if@firstcolumn\reversemarginpar\fi
   \else
     \if@firstcolumn\else\reversemarginpar\fi
   \fi
   \oldmarginnote{{\small\color{gray}#1}}
   \endgroup
}
\makeatother

\usepackage{expl3}
\ExplSyntaxOn
\cs_new_eq:NN \Rrepeat \prg_replicate:nn
\ExplSyntaxOff

\usepackage[noabbrev]{cleveref}

\newcommand{\sqlfont}{\fontsize{7}{9}\selectfont\sf}
\newcommand{\sqlfontsm}{\fontsize{5.5}{4.5}\selectfont\sf}
\newcommand{\sqlfontlg}{\fontsize{8}{9}\selectfont\sf}
\newcommand{\ssf}[1]{{{\sf #1}}}

\newenvironment{sqlquery}{
\vspace*{.4mm}
\sqlfont
\begin{tabbing}
\hspace*{\parindent}\=\hspace*{1.8cm}\=\hspace*{1.3cm}\kill}
{\end{tabbing}}

\newenvironment{sqlquerylg}{
\vspace*{1mm}
\sqlfontlg
\begin{tabbing}
\hspace*{\parindent}\=\hspace*{1.8cm}\=\hspace*{1.3cm}\kill}
{\end{tabbing}}

\newcommand{\SELECT}{\ssf{SELECT}\xspace}
\newcommand{\PACKAGE}[1]{\ssf{PACKAGE}$(#1)$\xspace}
\newcommand{\FROM}{\ssf{FROM}\xspace}
\newcommand{\REPEAT}{\ssf{REPEAT}\xspace}
\newcommand{\WHERE}{\ssf{WHERE}\xspace}
\newcommand{\SUCHTHAT}{\ssf{SUCH}~\ssf{THAT}\xspace}
\newcommand{\MAXIMIZE}{\ssf{MAXIMIZE}\xspace}
\newcommand{\MINIMIZE}{\ssf{MINIMIZE}\xspace}
\newcommand{\AS}{\ssf{AS}\xspace}
\newcommand{\AND}{\ssf{AND}\xspace}
\newcommand{\BETWEEN}{\ssf{BETWEEN}\xspace}
\newcommand{\WITHPROBABILITY}[1]{\ssf{WITH}~\ssf{PROBABILITY}#1\xspace}
\newcommand{\PROBABILITYOF}[1]{\ssf{PROBABILITY}~\ssf{OF}\ifthenelse{\equal{#1}{}}{}{~#1}\xspace}
\newcommand{\EXPECTED}[1]{\ssf{EXPECTED}\ifthenelse{\equal{#1}{}}{}{~#1}\xspace}

\newcommand{\SUM}[1]{\text{\ssf{SUM}}(#1)}
\newcommand{\COUNT}[1]{\text{\ssf{COUNT}}(#1)}

\newcommand{\paql}{\textsc{PaQL}\xspace}
\newcommand{\spaql}{\textsc{sPaQL}\xspace}
\newcommand{\ilp}{ILP\xspace}
\newcommand{\dilp}{DILP\xspace}
\newcommand{\silp}{SILP\xspace}
\newcommand{\stpr}{SP\xspace}

\newcommand{\sql}{SQL\xspace}
\newcommand{\spq}{SPQ\xspace}
\newcommand{\cplex}{CPLEX\xspace}
\newcommand{\tpch}{\mbox{TPC-H}\xspace}

\newcommand{\sdss}{SDSS\xspace}

\newcommand{\vfeasible}{feasible\xspace}
\newcommand{\vinfeasible}{infeasible\xspace}

\DeclarePairedDelimiter{\ceil}{\lceil}{\rceil}

\newcommand{\reals}{{\rm I\!R}}

\newcommand{\dattr}[1]{{\bf #1}}
\newcommand{\attr}[1]{\textbf{\textup{\textsc{#1}}}}
\newcommand{\rel}[1]{{\sf #1}}
\newcommand{\p}{{\sf P}\xspace}

\newcommand{\query}{\mathcal{Q}}

\newcommand{\expe}[1]{\mathbb{E}\left(#1\right)}
\newcommand{\expes}[1]{\mathbb{E}(#1)}
\newcommand{\expebg}[1]{\mathbb{E}\bigl(#1\bigr)}
\newcommand{\prob}[1]{\Pr\left(#1\right)}
\newcommand{\probs}[1]{\Pr({#1})}
\newcommand{\probbg}[1]{\Pr\bigl({#1}\bigr)}
\newcommand{\probbgg}[1]{\Pr\biggl({#1}\biggr)}
\newcommand{\indi}[1]{\mathds{1}\left(#1\right)}
\newcommand{\indis}[1]{\mathds{1}(#1)}
\newcommand{\indibg}[1]{\mathds{1}\bigl(#1\bigr)}
\newcommand{\indiBg}[1]{\mathds{1}\Bigl(#1\Bigr)}

\newtheorem{proposition}{Proposition}

\newtheorem{definition}{Definition}
\newtheorem{example}{Example}

\newcommand{\naive}{\textsc{Na\"{\i}ve}\xspace}
\newcommand{\sss}{\textup{\textsc{Summary\-Search}}\xspace}
\newcommand{\CSAsolve}{\textsc{\textup{CSA-Solve}}\xspace}

\newcommand{\sr}{\textsc{Sketch\-Refine}\xspace}

\newcommand{\SAA}{\textup{SAA}\xspace}
\newcommand{\CSA}{\textup{CSA}\xspace}

\newcommand{\Validate}[1]{\Call{Validate}{#1}\xspace}

\usepackage{algorithm}
\usepackage{algpseudocode}
\algnewcommand{\LineComment}[1]{\State {\footnotesize\color{gray}\(\triangleright\) #1}}

\algtext*{EndWhile}
\algtext*{Until}
\algtext*{EndIf}
\algtext*{EndFor}
\algtext*{EndProcedure}

\newcommand{\para}[1]{\smallskip\noindent{\bf #1}.\xspace}

\newcommand{\mass}{$^*$}
\newcommand{\nyu}{$^\circ$}

\newcommand{\authorsep}{\qquad}

\newcommand{\plot}[1]{\IfFileExists{#1}{\includegraphics[width=2.07cm]{#1}}{
\includegraphics[width=2.07cm]{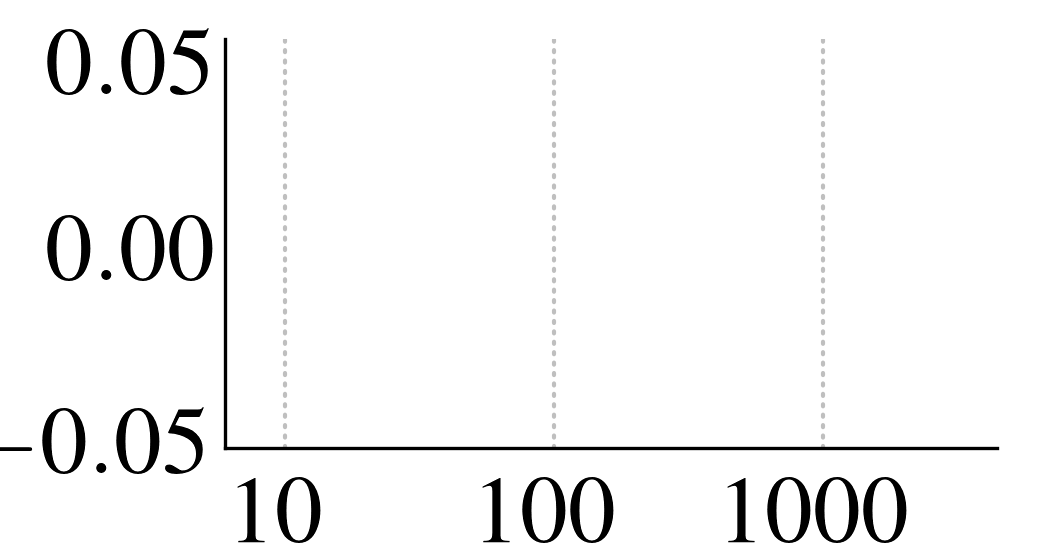}}}

\newif\ifshowappendix
\showappendixtrue

\copyrightyear{2020}
\acmYear{2020}
\setcopyright{acmlicensed}\acmConference[SIGMOD'20]{}{}{}

\begin{document}

\fancyhead{}

\title{Stochastic Package Queries in Probabilistic Databases$\textsuperscript{*}$}
\subtitle{Extended Version}

\author{Matteo Brucato\mass{}\authorsep{} Nishant Yadav\mass{}}
\author{Azza Abouzied\nyu{}\authorsep{} Peter J.\ Haas\mass{}\authorsep{} Alexandra Meliou\mass}
\affiliation{
\begin{tabular}{c c c}
    \mass\institution{University of Massachusetts Amherst} & \phantom{XXX} & \nyu\institution{NYU Abu Dhabi}\\
    \institution{\{matteo, nishantyadav, phaas, ameli\}@cs.umass.edu} && \institution{azza@nyu.edu}
\end{tabular}
}

\renewcommand{\shortauthors}{Brucato et al.}
\renewcommand{\authors}{Matteo Brucato, Nishant Yadav, Azza Abouzied, Peter J. Haas, and Alexandra Meliou}

\begin{abstract}
We provide methods for in-database support of decision making under uncertainty.
Many important decision problems correspond
to selecting a \emph{package} (bag of tuples in a relational database)
that jointly satisfy a set of constraints
while minimizing some overall cost function;
in most real-world problems, the data is uncertain.
We provide methods for specifying---via a \sql extension---and
processing \emph{stochastic package queries (\spq{s})}, in order to
solve optimization problems over uncertain data, right where the data resides.
Prior work in stochastic programming uses Monte Carlo
methods where the original stochastic optimization problem is
approximated by a large deterministic optimization problem
that incorporates many \emph{scenarios}, i.e.,
sample realizations of the uncertain data values.
For large database tables, however, a huge number of scenarios is required,
leading to poor performance and, often, failure of the solver software.
We therefore provide a novel \sss algorithm that,
instead of trying to solve a large deterministic problem,
seamlessly approximates it via a sequence of smaller problems
defined over carefully crafted \emph{summaries} of the scenarios
that accelerate convergence to a feasible and near-optimal solution.
Experimental results on our prototype system
show that \sss can be orders of magnitude faster than
prior methods at finding feasible and high-quality packages.
\end{abstract}

\maketitle

\vspace{10pt}
\section{Introduction} \label{sec:introduction}

\begin{figure*}
    \begin{minipage}[t]{.31\textwidth}
        \centering
        \rel{Stock\_Investments} (Table)
        {\small{
        \begin{tabular}{|c|cccc|}
            \hline
            \dattr{id} & \dattr{stock} & \dattr{price} & \dattr{sell\_in} & $\attr{Gain}$ \\
            \hline
            1 & AAPL          & 234           & 1 day            & ? \\
            2 & AAPL          & 234           & 1 week           & ? \\
            3 & MSFT          & 140           & 1 day            & ? \\
            4 & MSFT          & 140           & 1 week           & ? \\
            5 & TSLA          & 258           & 1 day            & ? \\
            6 & TSLA          & 258           & 1 week           & ? \\
            \hline
        \end{tabular}
        }}
    \end{minipage}
    \hspace{4mm}
    \begin{minipage}[t]{.32\textwidth}
        Stochastic Package Query (\spaql)\\
        \begin{sqlquery}
        \hspace*{2mm}\=\hspace*{1.6cm}\=\hspace*{1.3cm}\= \kill
        \SELECT\; \PACKAGE{\ast}\; \AS\; \rel{Portfolio} \\
        \FROM\; \rel{Stock\_Investments} \\
        \SUCHTHAT        \\
        \>$\SUM{\dattr{price}} \le 1000$ \AND \\
        \>$\SUM{\attr{Gain}} \ge -10$ $\WITHPROBABILITY{\ge 0.95}$ \\
        \MAXIMIZE $\EXPECTED{\SUM{\attr{Gain}}}$
        \end{sqlquery}
    \end{minipage}
    \begin{minipage}[t]{.32\textwidth}
        \centering
        \rel{Portfolio} (Package)\\
        {\small{
        \begin{tabular}{|c|cccc|}
            \hline
            \dattr{id} & \dattr{stock} & \dattr{price} & \dattr{sell\_in} & \attr{Gain} \\
            \hline
            3 & MSFT          & 140           & 1 day            & ? \\
            3 & MSFT          & 140           & 1 day            & ? \\
            6 & TSLA          & 258           & 1 week           & ? \\
            \hline
            \multicolumn{5}{c}{} \\
            \multicolumn{5}{c}{\small\it ``Buy 2 MSFT shares, sell them tomorrow.} \\
            \multicolumn{5}{c}{\small\it   Buy 1 TSLA share, sell it in 1 week''.}
        \end{tabular}
        }}
    \end{minipage}
    \caption{Example input table for the {\sc Financial Portfolio} (left),
    its stochastic package query expression in \spaql (center),
    and an example output package (right) with a description of its meaning\
    for the investor.
    Stochastic attributes (\attr{Gain}, in this example) are denoted in small caps and their
    values are unknown (shown by a question mark). Sample realizations of the uncertain ? values are generated by calls to VG functions.
    }
    \label{fig:portfolio}
\end{figure*}

Constrained optimization is central to decision making over a broad range of domains,
including
finance~\cite{jorion2007financial,hong2014monte},
transportation~\cite{clare2012air},
healthcare~\cite{geng2019addressing},
the travel industry~\cite{DeChoudhury2010},
robotics~\cite{du2011probabilistic}, and
engineering~\cite{bienstock2014chance}.
Consider, for example, the following very common investment problem.

\begin{example}[Financial Portfolio]
    \label{ex:portfolio}
    Given uncertain predictions for future stock prices based on financial models derived from historical data,
    an investor wants to invest $\$1{,}000$ in a set of trades
    (decisions on which stocks to buy and when to sell them)
    that will maximize the \emph{expected future gain},
    while ensuring that the \emph{loss} (if any) will be lower than $\$10$
    with probability at least $95\%$.
\end{example}

Suppose each row in a table contains a possible stock trade an investor can make:
whether to buy one share of a certain stock, and when to sell it back,
as shown in the left-hand side of Figure~\ref{fig:portfolio}.
The investor wants a ``package'' of trades---a subset of the input table, with possible repetitions
(i.e., multiple shares)---that is \emph{feasible}, in that it satisfies the given constraints
(total price at most \$1{,}000 and loss lower than \$10 with probability at least $95\%$), and \emph{optimal},
in that it maximizes an objective (expected future gain).
Although the current price of a stock is known---i.e.,
\dattr{price} is a \emph{deterministic} attribute---its future price,
and thus the gain obtained after reselling the stock, is \emph{unknown}.
In the input table, \attr{Gain} is a \emph{stochastic attribute}.
If the future gains were known, \Cref{ex:portfolio} would be a
``package query''~\cite{Brucato2018,brucato2014packagebuilder},
directly solvable as an \emph{Integer Linear Program} (\ilp) using off-the-shelf linear solvers
such as IBM \cplex~\cite{cplex}, and declaratively expressible in the Package Query Language (\paql).
Because \attr{Gain} is stochastic, the investor is solving a \emph{stochastic \ilp} instead.
In this paper, we introduce \emph{stochastic package queries} (SPQs),
a generalization of package queries that allows uncertainty in the data,
thereby allowing specification and solution of stochastic ILP problems.

We first introduce a simple language extension to \paql, called \spaql,
that allows easy specification of package queries with stochastic constraints and objectives.
We show the \spaql query for \Cref{ex:portfolio} in \Cref{fig:portfolio}.
The result of the query, on the right-hand side of the figure,
is a package that informs the investor about how
many trades to buy for each individual stock, and when to plan reselling them to the stock market.

Probabilistic databases~\cite{dalvi2007efficient,suciu2011probabilistic}
enable the representation of random variables in a database. 
The {\sc Financial Portfolio},
like many other real-world applications,
typically uses complex distributions to model uncertainty.
For instance, future stock prices are sometimes forecast using lognormal variates
based on ``geometric Brownian motion''~\cite{ross2014introduction}
using historical stock price data; alternatively, forecasts can incorporate complex stochastic predictive simulation or machine learning models.
For this reason, we base SPQs on the Monte Carlo probabilistic data model~\cite{jampani2008mcdb,JampaniXWPJH11},
which offers support for arbitrary distributions via user-defined \emph{variable generation} (VG) functions.
To generate a sample realization of the random variables in a database, the system calls the appropriate VG functions.
Whereas existing probabilistic databases excel at supporting \sql-like queries under uncertainty,
they do not support package-level optimization, and therefore cannot answer SPQs.
PackageBuilder~\cite{Brucato2018,brucato2014packagebuilder}, on the other hand,
only supports deterministic package queries and their translation into deterministic \ilp{s}.

The state of the art in solving stochastic \ilp{s} (\silp{s}) has been developed outside of the database setting,
in the field of~\emph{stochastic programming} (\stpr)~\cite{ahmed2008solving,CAMPI2009149,HOMEMDEMELLO201456}.
\stpr techniques approximate the given \silp
by a large deterministic ILP (\dilp) that simultaneously incorporates multiple \emph{scenarios}.
In a Monte Carlo database, a scenario is obtained by generating
a realization of every random variable in the table,
via a call to each associated VG function;
this procedure may be repeated multiple times,
generating a set of scenarios that are mutually independent and identically distributed (i.i.d.).
\Cref{fig:scenarios} shows an example of three possible scenarios for the input investment table
for \Cref{ex:portfolio}.
Roughly speaking, expectations in the \silp are approximated by averages over the scenarios
and probabilities by relative frequencies to form the \dilp,
which is then fed to a standard solver (e.g., \cplex).
The obtained solution approximates the true optimal solution for the \silp;
the more scenarios, the better the approximation.

The solution of the \dilp, however, may not be feasible with respect to the original \silp,
especially if the approximation is based on only a small number of scenarios
that do not well represent the true uncertainty distribution.
For example, a financial package obtained by using too few scenarios
might guarantee a loss less than $\$10$ with a probability
of only $65\%$, rather then $95\%$, incurring more risk than desired.

There is no practical way to know how many scenarios will be needed \emph{a priori};
existing theoretical a-priori bounds---see, e.g.,~\cite{luedtke2008sample}---are usually too conservative
to be usable when table sizes are large.
For example, if the \rel{Stock\_Investments} table contains $N{=}50{,}000$ rows,
then to guarantee that the \dilp solution is feasible for the \silp with merely $0.1\%$ probability
(which is really no guarantee at all), one would need $690{,}000$ scenarios,
resulting in a \dilp with $34.5$ \emph{billion} coefficients!
\stpr solutions must therefore be ``validated'' \emph{a posteriori}, using a much larger,
and out-of-sample, set of scenarios.
In \Cref{ex:portfolio}, for instance, we would generate, say, $10^6$ scenarios
and verify that the loss is less than $\$10$ in at least $95\%$ of them;
such validation is much faster than solving a \dilp with $10^6$ scenarios.

\begin{figure}
    \begin{center}
    \begin{minipage}[t]{.32\columnwidth}
        \centering
        Scenario 1
        {{
        \begin{tabular}{|c|cc|}
            \hline
            \dattr{id} & $\dots$ & \dattr{gain} \\
            \hline
            1 & $\dots$           & 0.1    \\
            2 & $\dots$           & 0.05   \\
            3 & $\dots$           & -0.2   \\
            4 & $\dots$           & 0.2    \\
            5 & $\dots$           & 0.1    \\
            6 & $\dots$           & -0.7    \\
            \hline
        \end{tabular}
        }}
    \end{minipage}
    \begin{minipage}[t]{.34\columnwidth}
        \centering
        Scenario 2
        {{
        \begin{tabular}{|c|cc|}
            \hline
            \dattr{id} & $\dots$ & \dattr{gain} \\
            \hline
            1 & $\dots$           & -0.2    \\
            2 & $\dots$           & -0.03   \\
            3 & $\dots$           & 0.5   \\
            4 & $\dots$           & 0.7    \\
            5 & $\dots$           & -0.7    \\
            6 & $\dots$           & -0.001    \\
            \hline
        \end{tabular}
        }}
    \end{minipage}
    \begin{minipage}[t]{.32\columnwidth}
        \centering
        Scenario 3
        {{
        \begin{tabular}{|c|cc|}
            \hline
            \dattr{id} & $\dots$ & \dattr{gain} \\
            \hline
            1 & $\dots$           & 0.01    \\
            2 & $\dots$           & 0.02   \\
            3 & $\dots$           & -0.1   \\
            4 & $\dots$           & -0.3    \\
            5 & $\dots$           & 0.2    \\
            6 & $\dots$           & 0.3    \\
            \hline
        \end{tabular}
        }}
    \end{minipage}
    \end{center}
    \caption{
    Three example scenarios for the \rel{Stock\_Investments} table, each
    showing only the \dattr{id}s and specific realizations for
    the stochastic attribute \attr{Gain}.}
    \label{fig:scenarios}
\end{figure}

The state-of-the-art algorithm thus works in a loop: the \emph{optimization phase} creates scenarios,
combines them into a \dilp, and computes a solution;
the \emph{validation phase} validates the solution against the out-of-sample scenarios.
If the solution is feasible on the validation scenarios (\emph{validation-feasible}),
the algorithm terminates, otherwise it creates more scenarios and repeats.
A solution that is validation-feasible is highly likely to be truly feasible for the original \silp.
Typically the ultimate number of scenarios used to compute the optimal solution to the \dilp
is astronomically smaller than the number prescribed by the conservative theoretical bounds
(though it is still large enough to be extremely computationally challenging).

Unfortunately, this process often breaks down in practice. Uncertainty increases with increasing table size,
and large tables typically need a huge number of scenarios to achieve feasibility.
Thus the validation phase repeatedly fails, and the scenario set---and hence
the \dilp---grows larger and larger until the solver is overwhelmed.
Even if the solver can ultimately handle the problem,
many ever-slower iterations may be required until validation-feasible solutions are found,
resulting in poor performance.

In this paper, we present an end-to-end system for SPQs,
seamlessly connecting \silp optimization with data management and stochastic predictive modeling.
Thus tasks related to efficiently storing data, maintaining consistency,
controlling access, and efficiently retrieving and preparing the data for analysis
can leverage the full power of a DBMS, while  avoiding the usual slow, cumbersome,
and error-prone analytics workflow where we read a dataset off of a database into main memory,
feed it to stochastic-prediction and optimization packages, and store the results back into the database.

We first introduce a \naive query evaluation algorithm,
which embodies the state-of-the-art optimization/validation technique outlined above,
and thoroughly discuss its drawbacks.
(Although the \naive technique is mentioned in the \stpr literature,
to our knowledge this is the first systematic implementation of the approach.)
We then introduce our new algorithm, \sss, that
is typically faster than \naive by orders of magnitude and can handle problems that cause \naive to fail.

Our key observation is that the randomly selected set of scenarios used to form the \dilp
during an iteration of \naive tend to be overly ``optimistic'',
leading the solver towards a seemingly good solution that ``in reality''---i.e.,
when tested against the validation scenarios---turns out to be infeasible.
This problem is also known as the ``optimizer's curse''~\cite{smith2006optimizer}.

To overcome the optimizer's curse, \sss replaces the large set of scenarios
used to form the \naive \dilp by a very small synopsis of the scenario set,
called a ``summary'', which results in a ``reduced'' \dilp that is much smaller than the \naive \dilp.
A summary is carefully crafted to be ``conservative''
in that the constraints in the reduced \dilp are harder to satisfy than the constraints in the \naive \dilp.
Because the reduced \dilp is much smaller than the \naive \dilp,
it can be solved much faster; moreover, the resulting solution is much more likely to be validation-feasible,
so that the required number of optimization/validation iterations is typically reduced.
Of course, if a summary is overly conservative, the resulting solution will be feasible,
but highly suboptimal.
Therefore, during each optimization phase, \sss implements a sophisticated search procedure
aimed at finding a ``minimally'' conservative summary;
this search requires solution of a sequence of reduced \dilp{s}, but each can be solved quickly.

Our experiments (Section~\ref{sec:experiments}) show that,
since its iterations are much faster than those of \naive,
\sss exhibits a large net performance gain even when the number of iterations is comparable;
typically, the number of iterations is actually much lower for \sss than for \naive,
further augmenting the performance gain.

\smallskip
In summary, the contributions of our paper are as follows.
\begin{itemize}[itemsep=5pt,wide,labelwidth=!,labelindent=0pt,leftmargin=\parindent]
    \item We extend the \paql language for deterministic package queries (itself an extension of \sql);
    the resulting language, \spaql (\Cref{sec:spaql}),
    allows specification of package queries with stochastic constraints and objectives.
    \item We provide a precise and concrete embodiment, the \naive algorithm,
    of the optimization/validation procedure suggested by the \stpr literature (Section~\ref{sec:monte-carlo}).
    \item We provide a novel algorithm, \sss, that is orders-of-magnitude faster than \naive,
    and that can solve \spq{s} that  require too many scenarios for \naive to handle.
    This is a significant contribution and fundamental extension to the known state-of-the-art
    in stochastic programming (\Cref{sec:summary-mc}).
	\item We present techniques that allow \sss to optimize its parameters automatically,
    and we provide theoretical approximation guarantees on the solution
    of \sss relative to \naive (Section~\ref{sec:optSum}).
    \item We provide a comprehensive experimental study,
    which indicates that \sss always finds validation-feasible solutions of high quality,
    even when \naive cannot, with dramatic speed-ups relative to \naive (Section~\ref{sec:experiments}).
\end{itemize}

\smallskip
Section~\ref{sec:related-work} discusses related work, and we conclude in Section~\ref{sec:conclusion-and-future}.
Our \spq techniques represent a significant step towards data-intensive decision making under uncertainty.

\section{Preliminaries} \label{sec:background}

Our work lies at the intersection of package queries, probabilistic databases, and stochastic programming.
In this section, we introduce some basic definitions from these areas that we will use throughout the paper.

\vspace{2pt}
\subsection{Deterministic Package Queries}\label{subsec:deterministic-package-queries}
A \emph{package} $\p$ of a relation $\rel{R}$ is a relation
obtained from $\rel{R}$ by inserting $m_{\p}(t)\ge 0$ copies of $t$ into $\p$
for each $t\in\rel{R}$;
here $m_{\p}$ is the \emph{multiplicity function} of $\p$.
The goal of a \emph{package query} is to specify $m_{\p}$,
and hence the tuples of the corresponding package relation.
A package query may include a \WHERE clause (tuple-level constraints),
a \SUCHTHAT clause (package-level constraints),
a package-level objective predicate and,
possibly, a \REPEAT limit, i.e., an upper bound on the number of duplicates of each tuple in the package.

A deterministic package query can be translated into an equivalent \emph{integer program}~\cite{Brucato2018}.
For each tuple $t_i \in \rel{R}$, the translation assigns a nonnegative integer decision variable $x_i$
corresponding to the multiplicity of $t_i$ in $\p$, i.e., $x_i=m_{\p}(t_i)$.
If the objective function and all constraints are linear in the $x_i$'s, the resulting integer program is an \ilp.
A cardinality constraint $\COUNT{\ast} = 3$ is translated into the \ilp constraint $\sum_{i=1}^{N} x_i = 3$.
A summation constraint $\SUM{\dattr{price}} \le 1000$ is translated into
$\sum_{i=1}^{N} t_i.\dattr{price} \,x_i \le 1000$;
this translation works similarly for other linear constraints and objectives.
A \REPEAT $l$ constraint is translated into bound constraints $x_i \le l + 1, \forall i\in[1..N]$.

\vspace{2pt}
\subsection{Monte Carlo Relations}
We use the Monte Carlo database model to represent uncertainty in a probabilistic database.
Uncertain values are modeled as random variables,
and a scenario (a deterministic realization of the relation) is generated by invoking
all of the associated VG functions for the relation.
In the simplest case, where all random variables are statistically independent,
each random variable has its own VG function;
in general, multiple random variables can share the same VG function,
allowing specification of various kinds of statistical correlations.
A Monte Carlo database system such as MCDB~\cite{jampani2008mcdb}
(or its successor, SimSQL~\cite{cai2013simulation}) facilitates specification of VG
functions as user-defined functions.
We assume that there exists a deterministic key column that is the same in each scenario,
so that each scenario contains exactly $N$ tuples for some $N\ge 1$
and  the notion of the ``$i$th tuple $t_i$" is well defined across scenarios.
For simplicity, we focus henceforth on the case where a database comprises a single relation.
Our results extend to Monte Carlo databases containing multiple (stochastic) base relations in which
the \spq is defined in terms of a relation obtained via a query over the base relations.

\subsection{Stochastic \ilp{s}}
The field of stochastic programming (\stpr)~\cite{shapiro2009lectures,kall1994stochastic}
studies optimization problems---selecting values of decision variables, subject to constraints,
to optimize an objective value---having uncertainty in the data.
We focus on \silp{s} with linear constraints and linear objectives that are deterministic,
expressed as expectations, or expressed as probabilities.
Probabilistic constraints are also called ``chance'' constraints in the \stpr literature.

\smallskip\noindent\textit{Linear constraints}.
Given random variables $\xi_1,\ldots,\xi_N$, decision variables $x_1,\ldots x_N$,
a real number $v \in \reals$, and a relation $\odot \in \{\le, \ge\}$,
a linear \emph{expectation constraint} takes the form $\expebg{\sum_{i=1}^{N} \xi_i x_i} \odot v$,
and a linear \emph{probabilistic constraint} takes the form $\probbg{\sum_{i=1}^{N} \xi_i x_i \odot v} \ge p$,
where $p\in [0,1]$.
We refer to $\sum_{i=1}^{N} \xi_i x_i \odot v$
as the \emph{inner constraint} of the probabilistic constraint,
and to $\sum_{i=1}^{N} \xi_i x_i$ as its \emph{inner function}.
Constraints of the form $\prob{\cdot} \le p$ can be rewritten in the aforementioned form
by flipping the inequality sign of the inner constraint and using $1-p$ instead.
If for constants $c_1,\ldots,c_N\in\reals$ we have $\probs{\xi_i=c_i}=1$ for $i\in[1..N]$,
then we obtain the deterministic constraint $\sum_{i=1}^{N} c_i x_i \odot v$
as a special case of an expectation constraint.

\smallskip\noindent\textit{Objective}.
Without loss of generality, we assume throughout that the objective has the canonical form
$\min_x\sum_{i=1}^{N} c_i x_i$ for deterministic constants $c_1,\ldots,c_N$.
Indeed, observe that an objective in the form of an expectation of a linear function
can be written in canonical form: $\min_x\expebg{\sum_{i=1}^{N} \xi_i x_i} = \min_x\sum_{i=1}^{N} \expe{\xi_i} x_i$,
and thus we take $c_i = \expe{\xi_i}$.
(This assumes that each expectation $\expe{\xi_i}$ is known or can be accurately approximated.)
We call $\sum_{i=1}^{N} \xi_i x_i$ the \emph{inner function} of the expectation.
Similarly, an objective in the form of a probability can be written in canonical form
using epigraphic rewriting~\cite{campi2018wait}.
For example, we can rewrite an objective of the form
$\min_x \probbg{\sum_{i=1}^{N} \xi_i x_i \odot v}$ in canonical form as $\min_{x,y} y$
and add a new probabilistic constraint $\probbg{\sum_{i=1}^{N} \xi_i x_i \odot v} \le y$.
Here $c_1=\cdots c_N=0$ and $y$ is an artificial decision variable
added to the problem with objective coefficient $c_y=1$.
Throughout the rest of the paper, we will primarily focus on techniques for minimization problems
with a nonnegative objective function;
the various other possibilities can be handled with suitable modifications
and are presented in~\Cref{sec:optimality}.

In our database setting, we assume for ease of exposition that, in a given constraint or objective,
each random variable $\xi_i$ corresponds to a random attribute value $t_i.\attr{A}$
for some real-valued attribute $\attr{A}$;
a different attribute can be used for each constraint,
and need not be the same as the attribute that appears in the objective.
Our methods can actually support more general formulations:
e.g., an expectation objective of the form $\min_x \expebg{\sum_{i=1}^N g(t_i)x_i}$,
where $g$ is an arbitrary real-valued function of tuple attributes;
constraints can similarly be generalized.
Note that this general form allows categorical attributes to be used in addition to real-valued attributes.

\section{Na\"{\i}ve \silp Approximation} \label{sec:monte-carlo}

\begin{algorithm}[t]
    \caption{Na\"{\i}ve Monte Carlo Query Evaluation}
    \label{algo:mc-evaluation}
    \begin{flushleft}{\small
        \begin{itemize}[noitemsep,topsep=0pt,leftmargin=*,label={}]
            \item $\query:$ A stochastic package query
            \item $\hat{M}:$ Number of out-of-sample validation scenarios (e.g., $10^{6}$)
            \item $M:$ Initial number of optimization scenarios (e.g., $100$)
            \item $m:$ Iterative increment to $M$ (e.g., $100$)
        \end{itemize}
        \textbf{output}: A \vfeasible package solution $x$, or failure (no solution).
        \begin{algorithmic}[1]
            \State $\mathcal{S} \gets \Call{GenerateScenarios}{\query,M}$ \Comment{\textit{Optimization scenarios}}
            \label{line:Ngeneration}
            \Repeat
                \State $\SAA_{\query,M} \gets \Call{FormulateSAA}{\query, \mathcal{S}}$
                \Comment{\textit{Approximate DILP}}
                \label{line:formulation}
                \State $x \gets \Call{Solve}{\SAA_{\query,M}}$
                \Comment{\textit{Solve \textnormal{SAA} with $M$ scenarios}}
                \label{line:Nsolve}
                \State $\hat{v}_x \gets \Validate{x, \query, \hat{M}}$
                \Comment{\textit{Validate $x$ using $\hat{M}$ scenarios}}
                \label{line:Nvalidate}
                \If{$\hat{v}_x$.is\_feasible}
                    \Comment{\textit{$x$ is \vfeasible}}
                    \State \Return $x$
                \EndIf

                \LineComment{\textit{Otherwise, use more optimization scenarios}}
                \State $\mathcal{S} \gets \mathcal{S} \cup \Call{GenerateScenarios}{\query,m}$
                \State $M \gets M + m$
                \label{line:Nincrement}
            \Until
        \end{algorithmic}
    }
    \end{flushleft}
\end{algorithm}

Recall that \naive is the first systematic implementation of the optimization/validation approach
mentioned in the \stpr literature.
The pseudocode is given as \Cref{algo:mc-evaluation}.
As discussed previously, the algorithm generates scenarios (line~\ref{line:Ngeneration}),
combines them into an approximating \dilp (line~\ref{line:formulation}),
solves the \dilp to obtain a solution $x$ (line~\ref{line:Nsolve}),
and then validates the feasibility of $x$
against a large number of out-of-sample validation scenarios (line~\ref{line:Nvalidate}).
The process is iterated, adding additional scenarios at each iteration (line~\ref{line:Nincrement})
until the validation phase succeeds.
We now describe these steps in more detail.

As discussed in the Introduction, the optimization phase for the \dilp can be very slow,
and often the convergence to feasibility requires so many optimize/validate iterations
that the \dilp becomes too large for the solver to handle, so that \naive fails.
Our novel \sss algorithm in \Cref{sec:summary-mc} uses ``summaries'' to speed up the optimization phase
and reduce the number of required iterations.

\vspace*{5pt}
\subsection{Sample-average approximation} \label{subsec:SAA}

As mentioned previously, we can generate a scenario by invoking all of the VG functions
for a table to obtain a realization of each random variable,
and can repeat this process $M$ times to obtain a Monte Carlo sample of $M$ i.i.d. scenarios.
In our implementation, \naive generates scenarios by seeding the random number generator
once for the entire execution, and accumulates scenarios in main memory.

We then obtain the \dilp from the original \silp by replacing the distributions of the random variables
with the empirical distributions corresponding to the sample.
That is, the probability of an event is approximated by its relative frequency in the sample,
and the expectation of a random variable by its sample average.
In the \stpr literature, this approach is known as \emph{Sample Average Approximation}
(\SAA)~\cite{ahmed2008solving,luedtke2008sample},
and we therefore refer to the \dilp for the stochastic package query $\query$ as $\SAA_{\query, M}$.

More formally, suppose that we have $M$ scenarios $S_1,\ldots,S_M$, each with $N$ tuples.
Recall that $t_i.\attr{A}$ denotes the random variable corresponding to attribute $\attr{A}$ in tuple $t_i$,
and denote by $s_{ij}.\attr{A} \in \reals$ the realized value of $t_i.\attr{A}$ in scenario~$S_j$.
Then each expected sum $\expebg{\sum_{i=1}^N t_i.\attr{A}\; x_i} = \sum_{i=1}^N \expe{t_i.\attr{A}} x_i$
is approximated by $\sum_{i=1}^{N} t_{i}.{\bar{\mu}_\attr{A}}\; x_i$,
where $t_{i}.\bar{\mu}_\attr{A}=(1/M)\sum_{j=1}^{M} s_{ij}.\attr{A}$.

To approximate a probabilistic constraint of the form
\begin{equation}\label{eq:probconst}
\probbgg{\sum_{i=1}^N t_i.\attr{A}\; x_i \odot v} \ge p,
\end{equation}
we add to the problem a new \emph{indicator variable},
${y_{j} \in \{0,1\}}$ for each scenario $j \in [1..M]$,
along with an associated \emph{indicator constraint}:
$y_{j} = \indi{\textstyle\sum_{i=1}^{N} s_{ij}.\attr{A}\; x_i \odot v}$,
where the indicator function $\indi{\cdot}$ equals 1 if the inner constraint is satisfied and equals 0 otherwise.
We say that solution $x$ ``satisfies scenario~$S_j$'' (with respect to the constraint) if and only if $y_{j} = 1$.
(Solvers like \cplex can handle indicator constraints.)
Finally, we add the following linear constraint over the indicator variables:
$\textstyle\sum_{j=1}^{M} y_{j} \ge \lceil p M \rceil$,
where $\lceil u\rceil$ is the smallest integer greater than or equal to $u$.
That is, we require that the solution $x$ satisfies at least a fraction $p$ of the $M$ scenarios.
The $\textsc{FormulateSAA}()$ function applies these approximations to create the \dilp $\SAA_{\query, M}$.

\para{Size complexity}
With $K$ constraints, the size of $\SAA_{\query, M}$,
measured with respect to the number of coefficients, is $\Theta(NMK)$:
we have $N$ coefficients for each expectation constraint and, for each probabilistic constraint,
$N+1$ coefficients (for $x_1,\ldots,x_N,y_j$) for each scenario.

\subsection{Out-of-sample validation} \label{subsec:validation}

After using $M$ scenarios to create and solve the \dilp $\SAA_{\query, M}$,
we check to see if the solution $x$ is \emph{validation-feasible}
in that it is a feasible solution for the \dilp $\SAA_{\query, \hat{M}}$
that is constructed using $\hat{M}\gg M$ out-of-sample scenarios.
When $\hat{M}$ is sufficiently large, validation feasibility is a proxy for true feasibility, i.e.,
feasibility for the original \silp;
commonly, $\hat{M} = 10^6$ or $10^7$.
This definition of validation-feasibility is simple, but widely accepted~\cite{luedtke2008sample}.
Although there are other, more sophisticated ways to use validation scenarios
to obtain confidence intervals on degree of constraint violation---see, e.g.,~\cite{campi2018wait}---these are
orthogonal to the scope of this paper.
Henceforth, we use the term ``feasibility'' to refer to ``validation feasibility'', unless otherwise noted.

In our implementation, during a precomputation phase,
we actually average $\hat M\gg M$ scenarios---the same number as the number of validation scenarios---to
estimate each $\expe{t_i.\attr{A}}$; 
we then append these  estimates, denoted $t_{i}.{\hat{\mu}_\attr{A}}$, to the table.
We do this because such averaging is typically very fast to execute,
and is space-efficient in that we simply maintain running averages.
Thus a solution $x$ returned by a solver is always \vfeasible for every expectation constraint,
and hence is \vfeasible overall if and only if,
for every probabilistic constraint of the form~\eqref{eq:probconst},
$x$ satisfies at least a fraction $p$ of the validation scenarios.
We can therefore focus attention on the probabilistic constraints, which are the most challenging.

The procedure $\Validate{x, \query, \hat{M}}$ checks the feasibility of $x$,
the solution to $\SAA_{\query, M}$;
we describe its operation on a single probabilistic constraint
$\probbg{\sum_{i=1}^N t_i.\attr{A}\;x_i\odot v}\ge p$,
but the same steps are taken independently for each probabilistic constraint.
It first seeds the system random number generator with a different seed than the one used
to generate the optimization scenarios.
For each $j\in[1..\hat{M}]$, it generates a realization $\hat{s}_{ij}.\attr{A}$
for each $t_i.\attr{A}$ such that $x_i > 0$
(i.e., for each tuple that appears in the solution package),
and computes the ``score'' $\sigma_j=\sum_{i:x_i>0} \hat{s}_{ij}.\attr{A}\;x_i$.
It then sets $y_j=\indis{\sigma_j\odot v}$.
After all scenarios have been processed, it computes $Y=\sum_{j=1}^{\hat M}y_j$
and declares $x$ to be feasible if $Y\ge \lceil p\hat M\rceil$.
The algorithm purges all realizations from main memory after each scenario has been processed,
and only stores the running count of the $y_j$'s,
allowing it to scale to an arbitrary number of validation scenarios.
Moreover, a package typically contains a realtively small number of tuples,
so only a small number of realizations need be generated.

\section{Summary-Based Approximation} \label{sec:summary-mc}

The \naive algorithm has three major drawbacks. (1)~The overall time to derive
a feasible solution to $\SAA_{\query, M}$ can be unacceptably long, since the
size of $\SAA_{\query, M}$ sharply increases as $M$ increases. (2)~It often
fails to obtain a feasible solution altogether---in our experiments, the
solver (\cplex) started failing with just a few hundred optimization scenarios.
(3)~\naive does not offer any guarantees on how close the objective value
$\omega$ of the solution $x$ to $\SAA_{\query, M}$ is to the true objective
value $\hat\omega$ of the solution $\hat x$ to the \dilp $\SAA_{\query, \hat
M}$ that is based on the validation scenarios. (Recall that we use
$\SAA_{\query, \hat M}$ as a proxy for the actual \silp.) A feasible
solution $x$ that \naive provides can be far from optimal.

Our improved algorithm, \sss, which we present in this section,
addresses these challenges by ensuring the efficient generation of \vfeasible results
through much smaller ``reduced'' \dilp{s} that each replace a large collection of $M$ scenarios
with a very small number $Z$ of scenario ``summaries'';
in many cases it suffices to take $Z=1$.
We call such a reduced \dilp a \emph{Conservative Summary Approximation} (\CSA),
in contrast to the much larger sample average approximation (\SAA) used by \naive.
The summaries are carefully designed to be more ``conservative''
than the original scenario sets that they replace:
the constraints are harder to satisfy, and thus the solver is induced to produce \vfeasible solutions faster.
\sss also guarantees that, for any user-specified approximation error $\epsilon \ge \epsilon_{\text{min}}$ (where $\epsilon_{\text{min}}$ is defined in \Cref{subsec:early-termination}),
if the algorithm returns a solution $x$,
then the corresponding objective value $\omega$ satisfies $\omega\le(1+\epsilon)\hat\omega$;
in this case we say that  $x$ is a \emph{$(1+\epsilon)$-approximate} solution.
(Recall that we focus on minimization problems with nonnegative objective functions;
the other cases are discussed in~\Cref{sec:optimality}.)

\vspace*{4pt}
\subsection{Conservative Summary Approximation} \label{subsec:ssa}

We first define the concept of an $\alpha$-summary, and then describe how $\alpha$-summaries are used to construct a \CSA.

\para{Summaries} Recall that a solution $x$ to $\SAA_{\query, M}$ \emph{satisfies} a scenario $S_j$ with respect to a probabilistic constraint of the form of Equation~\eqref{eq:probconst} if $y_j=\indibg{\sum_{i=1}^N s_{ij}.\attr{A}\;x_i\odot v}=1$, where
$s_{ij}.\attr{A}$ is the realized value of $t_i.\attr{A}$ in $S_j$.

\vspace*{4pt}
\begin{definition}[$\alpha$-Summary]
Let $\alpha \in [0,1]$.
An \emph{$\alpha$-summary} $S=\{s_i.\attr{A}:1\le i\le N\}$ of a scenario set $\mathcal{S}=\{S_1,\ldots,S_M\}$
with respect to a probabilistic constraint $C$ of the form~\eqref{eq:probconst} is a collection of $N$ deterministic
values of attribute \attr{A} such that if a solution $x$
\emph{satisfies} $S$
in that $\sum_{i=1}^N s_i.\attr{A}\;x_i\odot v$,
then $x$ satisfies at least $\lceil\alpha M\rceil$ of the scenarios
in $\mathcal{S}$ with respect to $C$.
\end{definition}
\vspace*{4pt}

Constructing an $\alpha$-summary, for $\alpha > 0$, is simple:
Suppose that the inner constraint of probabilistic constraint $C$ has the form $\sum_{i=1}^{N} t_i.\attr{A}\;x_i \ge v$.
Given any subset of scenarios $G(\alpha)\subseteq\mathcal{S}$ of size exactly $\lceil\alpha M\rceil$,
we define $S$ as the tuple-wise minimum over $G(\alpha)$:
\[
s_i.\attr{A} \coloneqq \min_{S_j \in G(\alpha)} s_{ij}.\attr{A}
\]

\begin{figure}
    \centering
    {\small
    \begin{tabular}{|c|cc|c|c|cc|c|c|cc|}
        \multicolumn{3}{c}{Scenario 1} &\multicolumn{1}{c}{}
        & \multicolumn{3}{c}{Scenario 3} & \multicolumn{1}{c}{}
        & \multicolumn{3}{c}{0.66-Summary}\\
        \cline{1-3}\cline{5-7}\cline{9-11}
           \dattr{id} & $\dots$ & \dattr{gain}
        && \dattr{id} & $\dots$ & \dattr{gain}
        && \dattr{id} & $\dots$ & \dattr{gain}\\
        \cline{1-3}\cline{5-7}\cline{9-11}
           1 & $\dots$           & 0.1
        && 1 & $\dots$           & 0.01
        && 1 & $\dots$           & 0.01\\
           2 & $\dots$           & 0.05
        && 2 & $\dots$           & 0.02
        && 2 & $\dots$           & 0.02   \\
           3 & $\dots$           & -0.2
        && 3 & $\dots$           & -0.1
        && 3 & $\dots$           & -0.2   \\
           4 & $\dots$           & 0.2
        && 4 & $\dots$           & -0.3
        && 4 & $\dots$           & -0.3   \\
           5 & $\dots$           & 0.1
        && 5 & $\dots$           & 0.2
        && 5 & $\dots$           & 0.1   \\
           6 & $\dots$           & -0.7
        && 6 & $\dots$           & 0.3
        && 6 & $\dots$           & -0.7   \\
        \cline{1-3}\cline{5-7}\cline{9-11}
    \end{tabular}
    }
    \vspace{-2mm}
    \caption{
    Using two out of the three scenarios of \Cref{fig:scenarios}, we derive a 0.66-summary.}
    \label{fig:summaryFinancial}
\end{figure}

\vspace*{4pt}
\begin{proposition}
    $S$ is an $\alpha$-summary of $\,\mathcal{S}$ with respect to $C$.
\end{proposition}
\begin{proof}
    Suppose $x$ satisfies $S$, i.e., $\sum_{i=1}^{N} s_i.\attr{A}\;x_i \ge v$.
    Then, for every scenario $S_j \in G(\alpha)$,
    ${\sum_{i=1}^{N} s_{ij}.\attr{A} x_i \ge \sum_{i=1}^{N} s_i.\attr{A} x_i \ge v}$.
    Since $|G(\alpha)| = \ceil{\alpha M}$, the result follows.
\end{proof}

\Cref{fig:summaryFinancial} illustrates an $\alpha$-summary for the three scenarios in \Cref{fig:scenarios}, where $\alpha=0.66$ and $G(\alpha)$ comprises scenarios~1 and 3. The summary is conservative in that, for any choice $x$ of trades, the gain under the summary values will be less than the gain under either of the two scenarios. Thus if we can find a solution that satisfies the summary, it will automatically satisfy at least scenarios~1 and 3. It might also satisfy scenario~2, and possibly many more scenarios, including unseen scenarios in the validation set. Indeed, if we are lucky, and in fact our solution satisfies at least $100p\%$ of the scenarios in the validation set, then $x$ will be feasible with respect to the constraint on \attr{Gain}.

Clearly, for an inner constraint with $\le$, the tuple-wise \emph{maximum} of $G(\alpha)$ yields an $\alpha$-summary. While there may be other ways to construct $\alpha$-summaries, in this paper we only consider minimum and maximum summaries, and defer the study of other, more sophisticated summarization methods to future work. Importantly, a summary need not coincide with any of the scenarios in $\mathcal{S}$; we are exploiting the fact that optimization and validation are decoupled.

\vspace{6pt}
\para{\CSA formulation}
A \CSA is basically an \SAA in which all probabilistic constraints are approximated using summaries
instead of scenarios.\footnote{As with the \SAA formulation, expectations are approximated as averages
over a huge number $\hat M$ of independent scenarios.}
The foregoing development implicitly assumed a single summary (with respect to a given probabilistic constraint $C$) for all of the $M$ scenarios in $\mathcal{S}$. In general, we use $Z$ summaries, where $Z\in[1..M]$. These are obtained by dividing $\mathcal{S}$ randomly into $Z$ disjoint partitions $\Pi_1,\ldots,\Pi_Z$, of approximately $M/Z$ scenarios each. Then the $\alpha$-summary $S_z=\{s_{iz}.\attr{A}:1\le i\le N\}$ for partition $\Pi_z$ is obtained by taking a tuple-wise minimum or maximum over scenarios in a subset $G_z(\alpha)\subseteq\Pi_z$, where $|G_z(\alpha)|=\ceil{\alpha|\Pi_z|}$.

For each probabilistic constraint $C$ of form~\eqref{eq:probconst},
we add to the \dilp a new indicator variable, $y_z \in \{0,1\}$, and an associated indicator constraint
$y_{z} := \indibg{\textstyle\sum_{i=1}^{N} s_{iz}.\attr{A}\; x_i \odot v}$. We say that solution $x$ ``satisfies summary $S_z$'' iff ${y_{z} = 1}$. We also add the linear constraint $\sum_{z=1}^{Z} y_{z} \ge \lceil p Z \rceil$, requiring at least $100p\%$ of the summaries to be satisfied.
We denote the resulting reduced \dilp by $\CSA_{\query, M, Z}$.

\para{Size complexity}
Assuming $K$ probabilistic constraints, the number of coefficients in $\CSA_{\query, M, Z}$ is $\Theta(NZK)$,
which is independent of $M$.
Usually, $Z$ takes on only small values, so that the effective size complexity is only $\Theta(NK)$.

Our results (\Cref{sec:experiments}) show that in most cases \sss finds good solutions with only one summary,
i.e., $Z=1$. Because $Z$ is small, the solution to $\CSA_{\query, M, Z}$ can be rapidly computed by a solver.
The \CSA formulation is also more robust to random fluctuations in the sampled data values,
and less prone to ``overfit'' to an unrepresentative set of scenarios obtained by luck of the draw.

An important observation is that as $Z$ increases, $\CSA_{\query, M, Z}$ approaches the $\SAA_{\query,M}$ formulation:
at $Z=M$ each partition will  contain exactly one scenario, which will also coincide with the summary for the partition.
Since $\CSA_{\query, M, Z}$ encompasses $\SAA_{\query,M}$,
we can always do at least as well as \naive with respect to the feasibility and optimality properties of our solution,
given $M$ scenarios.
We address the issue of how to choose $Z$, $\alpha$, and each $G_z(\alpha)$ below and in Section~\ref{sec:optSum},
and also discuss how to generate summaries efficiently.

\begin{algorithm}[t]
\caption{\sss Query Evaluation}
\label{algo:ss-evaluation}
\begin{flushleft}{\small
\begin{itemize}[noitemsep,topsep=0pt,leftmargin=*,label={}]
    \item $\query:$ A stochastic package query with $K$ probabilistic constraints
    \item $\query_0:$ $\query$ devoid of all probabilistic constraints
    \item $\hat{M}:$ Number of out-of-sample validation scenarios (e.g., $10^{6}$)
    \item $M:$ Initial number of optimization scenarios (e.g., $100$)
    \item $m:$ Iterative increment to $M$ (e.g., $100$)
    \item $z:$ Iterative increment to $Z$ (e.g., $1$)
    \item $\epsilon:$ User-defined approximation error bound, $\epsilon \ge \epsilon_{\text{min}}$
\end{itemize}
\textbf{output}: A \vfeasible package solution $x$, or failure (no solution).
\begin{algorithmic}[1]

    \LineComment{\textit{Solve probabilistically-unconstrained problem}}
        \State $x^{(0)} \gets \Call{Solve}{\SAA(\query_0, \hat{M})}$
  \label{line:defXinit}
    \State $Z = 1$ \Comment{Initial number of summaries}
    \Repeat
        \State $(x, \hat{v}_x) \gets \Call{\CSAsolve}{\query, x^{(0)}, M, Z}$

        \If{$\hat{v}_x$.is\_feasible \textbf{and} $\hat{v}_x$.upper\_bound $\le \epsilon$} 
            \State \Return $x$ \Comment{$x$ is feasible and $({1+\epsilon})$-approximate}
        \ElsIf{$\hat{v}_x$.is\_feasible \textbf{and} $Z < M$}
            \State $Z \gets Z + \min\{z, M-Z\}$ \Comment{Use more summaries}
        \Else
            \State $M \gets M + m$ \Comment{Use more scenarios}
        \EndIf
    \Until
\end{algorithmic}}
\end{flushleft}
\end{algorithm}

\subsection{Query Evaluation with \CSA} \label{subsec:ssa-evaluation}

Algorithm~\ref{algo:ss-evaluation} shows query evaluation with \sss.
The goal is to find a feasible solution whose objective value is as close as possible to $\hat\omega$,
the objective value of the \SAA based on the $\hat M$ validation scenarios.
In the algorithm, $\query_0$ denotes the \spq obtained from $\query$ by removing all of the probabilistic constraints.
At the first step, \sss computes $x^{(0)}$, the solution to the \dilp $\SAA_{\query_0, \hat{M}}$;
the only constraints are deterministic constraints and expectation constraints,
with the latter estimated from $\hat{M}$ scenarios in the usual way.
This corresponds to the ``least conservative'' solution possible,
and is effectively equivalent to solving a \CSA using summaries constructed with $\alpha = 0$,
because $0\%$ (i.e., none) of the scenarios are required to be satisfied.
For some problems, $x^{(0)}$ might have an infinite objective value,
in which case we simply ignore this solution and incrementally increase $\alpha$ until we find a finite solution.

Like \naive, the \sss algorithm starts with an initial number of optimization scenarios,
$M \ge 1$, and iteratively increments it while solutions are \vinfeasible.
In the optimization phase, the algorithm uses a \CSA formulation,
which replaces the $M$ real scenarios with $Z$ conservative summaries.
Initially, the algorithm uses $Z=1$, replacing the set of $M$ scenarios with a single summary. 
After feasibility is achieved for a solution $x$ with objective value $\omega_x$, the algorithm tries to check whether the ratio $\epsilon_x=(\omega_x-\hat\omega)/\hat\omega$
is less than or equal to the user-defined error bound $\epsilon$; although $\hat\omega$, and hence $\epsilon_x$, is unknown, we can conservatively check whether $\epsilon'_x\le \epsilon$, where $\epsilon'_x$ is an upper bound on $\epsilon_x$ that we develop in \Cref{subsec:early-termination}, 
If the solution is unsatisfactory, \sss increases $Z$, and iterates again.
The algorithm stops if and when a \vfeasible and $(1+\epsilon)$-approximate solution is found.
In practice, because of the conservative nature of summaries,
\sss typically finds \vfeasible solutions in drastically fewer iterations than \naive.

\section{Optimal Summary Selection} \label{sec:optSum}

The key component of \sss is \CSAsolve, described in this section.
With $M$ and $Z$ fixed, \CSAsolve finds the best \CSA formulation, i.e.,
the one having, for each constraint, the optimal value of $\alpha$ and the best set $G_z(\alpha)$ of scenarios for each summary. \CSAsolve thus determines the best solution $x$ achievable with $M$ scenarios and $Z$ summaries, and also computes
metadata $\hat{v}_x$ used by \sss for checking feasibility and optimality.

\vspace*{-5pt}
\subsection{\CSAsolve Overview} \label{subsec:ssa-search}

\Cref{algo:summary-mc} depicts the iterative process of \CSAsolve:
at each iteration $q$ it produces a solution $x^{(q)}$ to a problem $\CSA_{\query, M, Z}$
based on an $\alpha_k^{(q)}$-summary for each constraint $C_k$.
Initially, $\alpha^{(0)}_k=0$ for all $k$,
and thus the solution to $\CSA_{\query, M, Z}$ is simply $x^{(0)}$,
which has already been computed by \sss prior to calling \CSAsolve.
Then \CSAsolve stops in two cases:
(1)~if it finds a \vfeasible $(1+\epsilon)$-approximate solution;
(2)~if it enters a cycle, producing the same solution twice with the same $\alpha_k$ values.
In case (2), it returns the ``best'' solution found so far:
if one or more feasible solutions have been found, it returns the one with the best objective value,
otherwise it returns an infeasible solution, and \sss will increase $M$ in its next iteration.

\vspace*{4pt}
\subsection{Choosing $\boldsymbol{\alpha}$} \label{subsec:choosing}
Larger $\alpha$ leads to more conservative $\alpha$-summaries,
as we take the tuple-wise minimum (or maximum) over more and more scenarios.
Thus a high value of $\alpha$ increases the chances of finding a \vfeasible solution.
On the other hand, if the constraints are more restrictive than necessary,
then the solution can have a seriously suboptimal objective value
because we are considering fewer candidate solutions, possibly missing the best ones.
Thus, \CSAsolve seeks the minimally conservative value of $\alpha$ that will suffice.

How can we measure the true conservativeness of $\alpha$ with respect to a constraint
$C \coloneqq \probs{\sum_{i=1}^{N} t_i.\attr{\emph{A}}\;x_i \odot v} \ge p$?
As discussed previously, the solution $x$ to a formulation $\SAA_{\query, M}$
based on $\alpha$-summaries is guaranteed to satisfy at least $100\alpha\%$ of the $M$ optimization scenarios,
but the actual true probability of satisfying the constraint---or more pragmatically,
the fraction of the $\hat{M}$ validation scenarios satisfied by $x$---will usually differ from $\alpha$.
Thus, we look at the difference between the fraction of validation scenarios satisfied by $x$
and the target value~$p$.
We call this difference the \emph{$p$-surplus}, and define it as:
\[
    r = r(\alpha) \coloneqq \biggl\{(1/\hat{M})
    \sum_{j=1}^{\hat{M}}\indiBg{\sum_{i=1}^{N} \hat{s}_{ij}.\attr{\emph{A}}\;x_i \odot v}\biggr\} - p
\]
We expect the function $r(\alpha)$ to be increasing in $\alpha$ with high probability.

Observe that $x$ essentially satisfies the constraint
$C' \coloneqq \probs{\sum_{i=1}^{N} t_i.\attr{A}\;x_i \odot v} \ge p + r$.
Clearly, if $r<0$, then $x$ is infeasible for constraint $C$, whereas if $r>0$,
then $x$ satisfies the inner constraint with a probability that exceeds $p$,
and so is conservative and therefore likely suboptimal.
Thus the optimal value $\alpha^*$ satisfies $r(\alpha^*)=0$.
Solutions that achieve zero $p$-surplus may be impossible to find,
and therefore \CSAsolve tries to choose $\alpha=(\alpha_1,\ldots,\alpha_K)$ to minimize the $p$-surplus for each of the K constraints,
while keeping it nonnegative. The search space is finite (hence the possibility of cycles) since $\alpha_k\in\{Z/M,2Z/M,\ldots,1\}$ for $k\in[1..K]$.

At each iteration $q$, \CSAsolve updates $\alpha^{(q-1)}$ to $\alpha^{(q)}$,
creates the corresponding \CSA problem, and produces a new solution $x^{(q)}$.
For simplicity and ease of computation, our initial implementation updates
each $\alpha_k^{(q)}$ individually by fitting a smooth curve $R^{(q)}_k(\alpha_k)$ to the historical points
$(\alpha_k^{(0)},r_k^{(0)}),\ldots,(\alpha_k^{(q-1)},r_k^{(q-1)})$ and then solving the equation $R^{(q)}_k(\alpha_k)=0$.
In our experiments, we observed that (1) fitting an arctangent function provides the most accurate predictions and (2)
this artificial decoupling with respect to the constraints yields effective summaries;
we plan to investigate other methods for jointly updating $(\alpha_1^{(q-1)},\ldots,\alpha_K^{(q-1)})$.

\begin{algorithm}[t]
    \caption{\sc \CSAsolve} \label{algo:summary-mc}
    \begin{flushleft}{\small
    \begin{itemize}[noitemsep,topsep=0pt,leftmargin=*,label={}]
        \item $\query:$ A stochastic package query with $K$ probabilistic constraints
        \item $x^{(0)}:$ Solution of probabilistically-unconstrained problem
        \item $M:$ Number of optimization scenarios
        \item $Z:$ Number of summaries, $1 \le Z \le M$
        \item $\epsilon:$ User-defined approximation error bound, $\epsilon \ge \epsilon_{\text{min}}$
    \end{itemize}
    \textbf{output}: A \vfeasible and $({1+\epsilon})$-approximate solution, or an \vinfeasible solution
    \begin{algorithmic}[1]
        \State $q \gets 0$    \Comment{\textit{Iteration count}}
        \State $\mathcal{H} \gets \emptyset$\Comment{\textit{Initialize validation history}}
        \State $\alpha^{(q)}=(\alpha_1^{(q)},\ldots,\alpha_K^{(q)}) \gets (0,\ldots,0)$ \Comment{\textit{Initial conservativeness}}

        \Repeat
        \LineComment{\textit{If entered a cycle, return best solution from history}}
        \If{$(x^{(q)},\alpha^{(q)})\in \mathcal{H}$}
          \State
          \Return $\Call{Best}{\{x : (x,\alpha) \in \mathcal{H}\}}$
        \EndIf
        \State $\mathcal{H}\gets \mathcal{H}\cup\{(x^{(q)},\alpha^{(q)})\}$  \Comment{\textit{Update validation history}}
        \State $\hat{v}^{(q)} \gets \Call{Validate}{x^{(q)}, \query, \hat{M}}$
               \Comment{\textit{Validate \& compute metadata}} \label{line:validation}
        \State $\epsilon^{(q)} \gets \hat{v}^{(q)}.\text{upper\_bound}$
              \Comment{\textit{Validation upper bound on $\epsilon$}}
        \For{$k = 1, \dots, K$}
          \State $r_{k}^{(q)} \gets \hat{v}^{(q)}_k.\text{surplus}$
          \Comment{\textit{Validation $p$-surplus}}
         
        \EndFor

        \LineComment{\textit{Termination with \vfeasible $({1+\epsilon})$-approximate solution}}
        \If{$\epsilon^{(q)} \le \epsilon$ \textbf{and} $\forall k: r_{k}^{(q)} \ge 0$}
          \label{line:earlyTerm}
          \State \Return $(x^{(q)}, \hat{v}^{(q)})$
        \EndIf

        \State $q \gets q + 1$ \Comment{Iterate again with a new set of summaries}
        \State $\alpha^{(q)} \gets$ \Call{GuessOptimalConservativeness}{$\mathcal{H}$}
        \For{$k = 1, \dots, K$}
          \State $\tilde{S}_k \gets \Call{Summarize}{x^{(q)}, \alpha_{k}^{(q)}, C_k, \mathcal{H}}$
        \EndFor

        \State $\CSA_{\query,M,Z} \gets \Call{FormulateSAA}{\query, \{\tilde{S}_1, \dots, \tilde{S}_K\}}$
        \State $x^{(q)} \gets \Call{Solve}{\CSA_{\query,M,Z}}$
        \Until
    \end{algorithmic}}
    \end{flushleft}
\end{algorithm}

\subsection{Choosing $\mathbf{G_z}$} \label{subsec:choosing-G}
So far, we have assumed that the subset $G_z(\alpha_k^{(q)})$
used to build the summary is \emph{any} set containing $n_k^{(q)}=\ceil{\alpha_k^{(q)}|\Pi_z|}$ scenarios.
\sss employs a simple greedy heuristic to determine $G_z(\alpha_k^{(q)})$:
it chooses the $n_k^{(q)}$ scenarios that produce the summary most likely to keep the previous solution
feasible in the current iteration, so that the new solution will likely have a higher objective value.
For an inner $\ge$ ($\le$) constraint, this is achieved by sorting the scenarios in $\Pi_z$ according to their
``scenario score'' $\sum_{i=1}^{N} s_{ij}.\attr{A}\;x_i^{(q-1)}$ and taking the first $n_k^{(q)}$
in descending (ascending) order.

\subsection{Approximation Guarantees} \label{subsec:early-termination}
If $x^{(q)}$ is \vfeasible, \sss can terminate if it can determine that $x^{(q)}$ is
$(1 + \epsilon)$-approximate relative to the optimal \vfeasible solution $\hat{x}$ based on the validation scenarios,
i.e., that $\omega^{(q)} \le ({1+\epsilon})\hat{\omega}$,
where $\omega^{(q)}$ and $\hat{\omega}$ are the objective values for $x^{(q)}$ and $\hat{x}$, respectively,
and $\epsilon$ is an accuracy parameter specified by the user.
Without loss of generality, we assume below that the objective function is an expectation;
should the objective be deterministic, nesting it within an expectation does not change its value.

This termination check proceeds as follows.
During the $q$th iteration of \sss, the function \Validate{$x^{(q)}, \query, \hat{M}$}
computes  $p$-surplus values $r_1^{(q)},\ldots,r_K^{(q)}$,
one for each probabilistic constraint in the query.
Further, it computes $\epsilon^{(q)}$ (as defined below).
We show below that if $\epsilon^{(q)} \le \epsilon$ and $\forall k: r_k^{(q)} \ge 0$,
then $x^{(q)}$ is a \vfeasible $({1+\epsilon})$-approximate solution,
and \sss can immediately return $x^{(q)}$ and terminate.
As usual, we focus on minimization problems with nonnegative objective values,
and take the optimal solution $\hat x$ and objective value $\hat\omega$ of $\SAA_{\query,\hat M}$
as proxies for those of the original \silp.
We start with the following simple but important result.

\begin{proposition}[General Approximation Guarantee]
    \label{th:unsupp-approx-mc}
    Let $\epsilon\ge 0$ and let $\underline{\omega}$
    be a positive constant such that  $\underline{\omega}\le \hat\omega$.
    Set $\epsilon^{(q)}=(\omega^{(q)}/\underline{\omega})-1$.
    If $\epsilon^{(q)} \le \epsilon$,
    then $\omega^{(q)} \le {(1 + \epsilon)} \hat\omega$.
\end{proposition}
\begin{proof}
    Suppose that $\epsilon^{(q)} \le \epsilon$.
    Since $\hat\omega/\underline{\omega}\ge 1$, we have
    \[
        \omega^{(q)}\le \Bigl(\frac{\hat\omega}{\underline{\omega}}\Bigr)\omega^{(q)}
        =\biggl(1+\Bigl(\frac{\omega^{(q)}}{\underline{\omega}}-1\Bigr)\biggr)\hat\omega
        = \bigl(1+\epsilon^{(q)}\bigr)\hat\omega\le (1+\epsilon)\hat\omega,
    \]
and the result follows.
\end{proof}

We obtain a specific formula for $\epsilon^{(q)}$
by choosing a specific bound $\underline{\omega}$.
Clearly, we would like to choose $\underline{\omega}$ as large as possible,
since this maximizes the likelihood that $\epsilon^{(q)} \le \epsilon$.
One simple choice that always works is to set $\underline{\omega}=\omega^{(0)}$,
where $\omega^{(0)}$ is the objective value of the \SAA problem corresponding to the original \silp
but with all probabilistic constraints removed---see line~\ref{line:defXinit} of Algorithm~\ref{algo:ss-evaluation}.
If all random variables are lower-bounded by a constant $\underline{s} > 0$ and the size of any feasible package is lower-bounded by a constant $\underline{l} > 0$,
then $\sum_{i=1}^N\hat{s}_{ij}.\attr{A}\;x_i\ge  \underline{s}\underline{l}$, $\forall j\in[1..\hat M]$, so that
\[
    \hat\omega =\frac{1}{\hat M}\sum_{j=1}^{\hat M}\sum_{i=1}^N\hat{s}_{ij}.\attr{A}\;\hat{x}_i
    \ge \frac{1}{\hat M}\sum_{j=1}^{\hat M}\underline{s}\underline{l}
    = \underline{s}\underline{l},
\]
which yields an alternative lower bound.
Yet another bound can be sometimes obtained by exploiting the relation of the constraints to the objective.

\begin{definition}[Objective-Constraint Interaction]\label{def:supportiveness}
    Let the objective be $\min\, \expes{\textstyle\sum_{i=1}^{N} \xi_i x_i}$, for random variables $\{\xi_i\}_{i\in[1..N]}$.
    The objective is said to be \emph{supported} by a constraint of the form
    $\probbg{\sum_{i=1}^{N} \xi_i x_i \le v} \ge p$ and \emph{counteracted} by a constraint of the form
    ${\probs{\textstyle\sum_{i=1}^{N} \xi_i x_i \ge v} \ge p}$. All other forms of constraint
    are said to be \emph{independent} of the objective.
\end{definition}

Intuitively, a supporting probabilistic constraint ``supports'' the objective function
in the same ``direction'' of the optimization ($\le$ for minimization, $\ge$ for maximization),
whereas a counteracting constraint goes against the optimization.
If there exists a counteracting constraint with $v \ge 0$, it can be shown (\Cref{sec:optimality}) that
$\hat\omega\ge pv$.

Finally, we take $\underline\omega$ to be the maximum of all applicable lower bounds.
Similar formulas can be derived for other possible cases---maximization problems, negative objective values, and so on;
see \Cref{sec:optimality}.

Note that if $(\hat\omega/\underline{\omega})-1>\epsilon$, then $\epsilon^{(q)}=(\omega^{(q)}/\underline{\omega})-1>\epsilon$, $\forall q\ge 0$,
so that \sss cannot terminate with a \vfeasible $(1+\epsilon)$-approximate solution.
To avoid this problem, we require that $\epsilon\ge\epsilon_{\text{min}}$,
where $\epsilon_{\text{min}}=(\overline\omega/\underline\omega)-1$.
Here  $\overline\omega$ is any upper bound on $\hat\omega$.
It can be shown, for example, that if
(1)~all random variables are upper-bounded by a constant $\overline{s} > 0$,
(2)~the size of any feasible package is upper-bounded by a constant $\overline{l} > 0$, and
(3)~there exists a supporting constraint with $v \ge 0$,
then $\hat\omega\le v+(1-p)\bar s\bar l$; see \Cref{sec:optimality}.
If we have available a feasible solution $x$ with objective value $\omega_x$,
then we can take $\overline{\omega}=\omega_x$.
We choose $\overline{\omega}$ to be the minimum of all applicable bounds.

\vspace{6pt}
\subsection{Implementation Considerations}\label{subsec:implementation}

We now discuss several implementation optimizations.

\para{Efficient summary generation} \label{subsec:scenario-construction}
Recall that summarization has two steps:
(1)~computing the scenario scores to sort scenarios by the previous solution, and (2)~computing the tuple-wise minimum (or maximum) of the first $\alpha\%$ of the scenarios in sorted order.
The fastest way to generate an $\alpha$-summary is if all $M$ scenarios are generated and kept in main memory at all times.
In this case, computing the tuple-wise minimum (or maximum) is trivial.
However, the $\Theta(MNK)$ memory requirement for this may exceed the memory limits if $M$ is large.
We devise two possible strategies for memory-efficient summary generation with optimal $\Theta(NZK)$ space complexity:
\emph{tuple-wise summarization} and \emph{scenario-wise summarization}.
Tuple-wise summarization uses a unique random number seed for each tuple ($i=1,\dots,N$)
and it generates all $M$ realizations, one tuple at a time.
Scenario-wise summarization uses a unique seed for each scenario ($j=1,\dots,M$),
and it generates one realization for all tuples, one scenario at a time.

With tuple-wise summarization, sorting the scenario only requires $\Theta(PM)$ time,
where $P=\sum_{i=1}^{N}x_i$ is the size of the current package;
usually, $P \ll N$.
However, generating the summaries is more costly, as it requires $\Theta(NM)$ time,
as all $M$ realizations must be constructed for all $N$ tuples.
The total time is $\Theta(M(P+N))$.
With scenario-wise summarization, generating summaries has lower time complexity of $\Theta(\alpha N M)$,
as it only generates scenarios in $G_z(\alpha)$, but sorting has higher complexity $\Theta(NM)$,
with total time $\Theta(NM (\alpha + 1))$.

It follows that if $\alpha \ge P/N$, tuple-wise summarization is generally faster than scenario-wise summarization.
However, other factors may affect the runtime, e.g., some random number generators,
such as Numpy, generate large quantities of random numbers faster if generated in bulk using a single seed.
In this case, tuple-wise summarization may suffer considerably in the summary generation phase,
as it needs to re-seed the random number generator for each tuple.
In our experiments, we observed that tuple-wise summarization is better when the input table is relatively small,
but worse than scenario-wise for larger tables.
In general, a system should implement both methods and test the two in situ.

\para{Convergence acceleration}
When $\alpha_k^{(q)}$ is obtained by \emph{decreasing} $\alpha_k^{(q-1)}$,
the solution $x^{(q-1)}$ typically is feasible,
and our goal is for $x^{(q)}$ to strictly improve in objective value.
\CSAsolve achieves this by slightly modifying the generation of summaries
in order to ensure that the previous solution is still feasible for the next \CSA problem.
This is done by using the tuple-wise maximum (instead of minimum)
in the summary generation for all tuples $t_i$ such that $x^{(q-1)}_i > 0$ (tuples in the previous solution).
For all other tuples, we set the summary as usual.
We have found that ensuring monotonicity of the objective values promotes faster convergence.

\section{Experimental Evaluation} \label{sec:experiments}

\setlength\tabcolsep{1.5pt} 
\renewcommand{\tabularxcolumn}[1]{>{\centering\arraybackslash}m{#1}}

\begin{figure*}
    \begin{mdframed}[linecolor=gray,innerleftmargin=10,innerrightmargin=10,innertopmargin=1,innerbottommargin=1]
    \centering{\scriptsize{
    \begin{tabularx}{\textwidth}{XX}
    \textbf{\naive}
    \includegraphics[width=5pt]{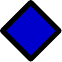}
    &
    \textbf{\sss}
    \includegraphics[width=5pt]{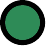}
    \end{tabularx}}}
    \end{mdframed}
    \vspace*{8pt}
    \scriptsize{
    \centering
    \begin{tabularx}{1\textwidth}{m{0.2cm}m{0.1cm}XXXXXXXX}
    & &
    \hspace{1.3em}Q1 & \hspace{1.3em}Q2 & \hspace{1.3em}Q3 & \hspace{1.3em}Q4 &
    \hspace{1.3em}Q5 & \hspace{1.3em}Q6 & \hspace{1.3em}Q7 & \hspace{1.3em}Q8\\
    \rotatebox{90}{\textbf{Galaxy}} &
    \rotatebox{90}{time (s)} &
    \plot{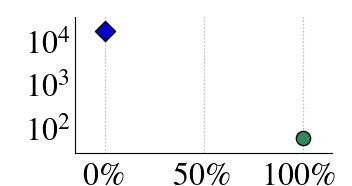}&
    \plot{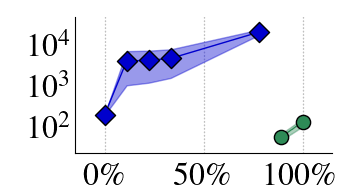}&
    \plot{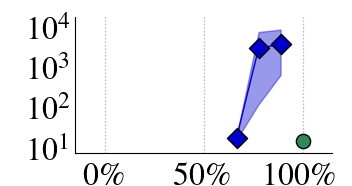}&
    \plot{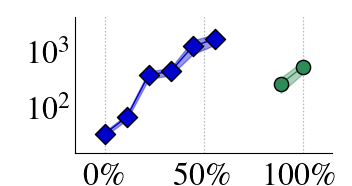}&
    \plot{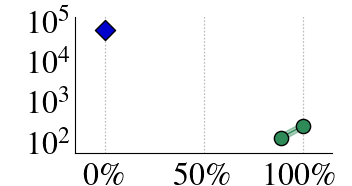}&
    \plot{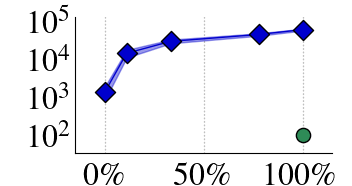}&
    \plot{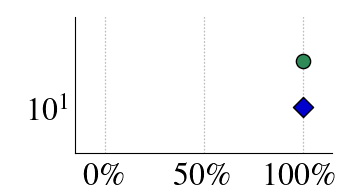}&
    \plot{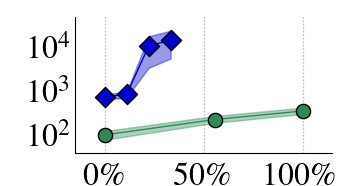}
    \end{tabularx}}
    \centering\scriptsize{
    \begin{tabularx}{1\textwidth}{m{0.2cm}m{0.1cm}XXXXXXXX}
    \rotatebox{90}{\textbf{Portfolio}} &
    \rotatebox{90}{\;\;time (s)}&
    \plot{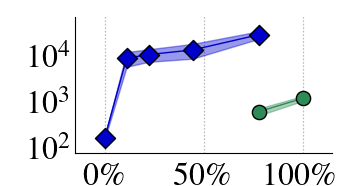}&
    \plot{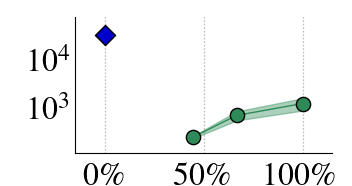}&
    \plot{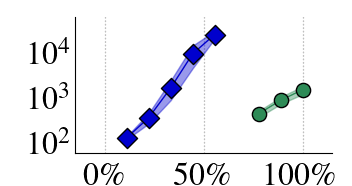}&
    \plot{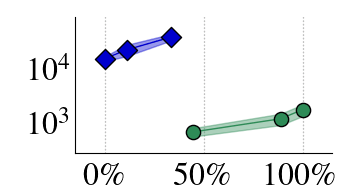}&
    \plot{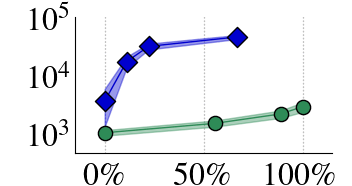}&
    \plot{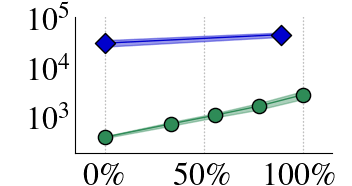}&
    \plot{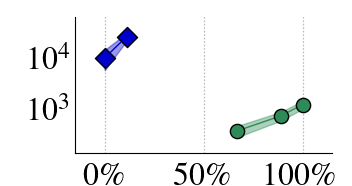}&
    \plot{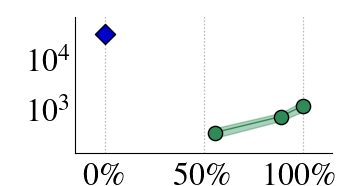}
    \end{tabularx}}
    \centering\scriptsize{
    \begin{tabularx}{1\textwidth}{m{0.2cm}m{0.1cm}XXXXXXXX}
    \rotatebox{90}{\textbf{\tpch}} &
    \rotatebox{90}{time (s)}&
    \plot{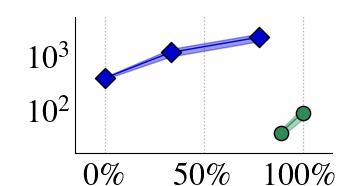}&
    \plot{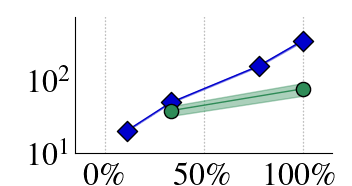}&
    \plot{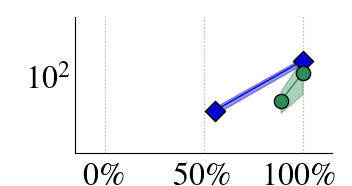}&
    \plot{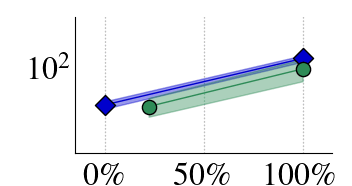}&
    \plot{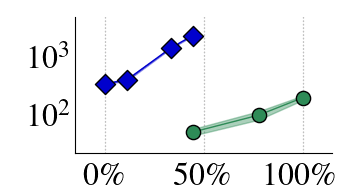}&
    \plot{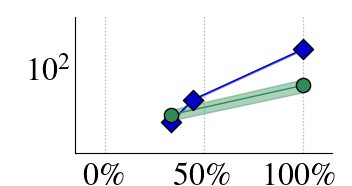}&
    \plot{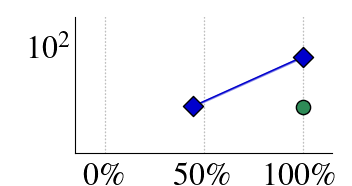}&
    \plot{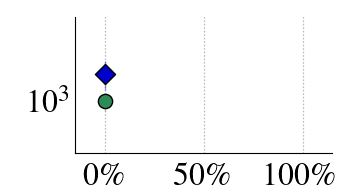}
    \\ & &
    \hspace{1.2em}\textbf{Feasibility Rate} &
    \hspace{1.2em}\textbf{Feasibility Rate} &
    \hspace{1.2em}\textbf{Feasibility Rate} &
    \hspace{1.2em}\textbf{Feasibility Rate} &
    \hspace{1.2em}\textbf{Feasibility Rate} &
    \hspace{1.2em}\textbf{Feasibility Rate} &
    \hspace{1.2em}\textbf{Feasibility Rate} &
    \hspace{1.2em}\textbf{Feasibility Rate}
    \end{tabularx}}
    \caption{
    End-to-end results of \sss vs. \naive.
    Plotting the average time (and 95\% confidence intervals) to reach 100\% feasibility rate.
    Of the 23 feasible queries (\tpch Q8 is infeasible), \sss always reaches 100\% feasibility rate, while \naive in only 7 queries.
    In 15 queries, when \sss succeeds, \naive is still at 0\% feasibility.
    \sss can be orders of magnitude faster even when both reach 100\% feasibility.
    }
    \label{fig:time-to-feasible}
    \vspace{18pt}
\end{figure*}

In this section, we present an experimental evaluation of our techniques
for stochastic package queries on three different domains where uncertainty naturally arises:
noise in sensor data, uncertainty in future predictions, uncertainty due to data integration~\cite{dong2009data}.
Our results show that:
(1)~\sss is always able to find \vfeasible solutions,
while \naive cannot in most cases---when both \sss and \naive can find \vfeasible solutions,
\sss is often faster by orders of magnitude;
(2)~The packages produced by \sss are of high quality (low empirical approximation ratio),
sometimes even better than \naive when they both produce feasible solutions;
(3)~Increasing $M$, the number of optimization scenarios, helps \sss find \vfeasible solutions,
and the value of $M$ required by \sss to start producing \vfeasible solutions is much smaller than \naive,
explaining the orders of magnitude improvement in running time;
(4)~Increasing $Z$, the number of summaries, helps \sss find higher-quality solutions;
(5)~Increasing $N$, the number of input tuples, impacts the running time of both algorithms,
but \sss is still orders of magnitude faster than \naive,
and finds \vfeasible solutions with better empirical approximation ratios than \naive.

\vspace*{5pt}
\subsection{Experimental Setup} \label{subsec:experimental-setup}
We now describe the software and runtime environment, and the three workloads we used in
the experiments.

\para{Environment}
We implemented our methods in Python 2.7, used Postgres 9.3.9 as the underlining DBMS,
and IBM \cplex 12.6 as the \ilp solver.
We ran our experiments on servers equipped with two 24 2.66GHz cores, 15GB or RAM, and a 7200 RPM 500GB hard drive.

\para{Datasets and queries}
We constructed three workloads:

\noindent\emph{Noisy sensor measurements}:
The Galaxy datasets vary between $55{,}000$ and $274{,}000$ tuples,
extracted from the Sloan Digital Sky Survey (\sdss)~\cite{sdss_dr12}.
Each tuple contains the color components of a small portion of the sky as read by a telescope.
We model the uncertainty in the telescope readings as Gaussian or Pareto noise.

\smallskip\noindent\emph{Financial predictions}:
The Portfolio dataset contains 6,895 stocks downloaded from Yahoo Finance~\cite{EmptyId-10}.
The initial price of each stock is set according to its actual value on January 2, 2018,
and future prices are generated according to a geometric Brownian motion.
We consider selling stocks in one day or in one week, 
as in \Cref{fig:portfolio}; the dataset for the short-term (resp., long-term) trades contains 14,000 (resp., 48,000) tuples.
For each prediction type, we also extracted a subset corresponding to the $30\%$ most volatile stocks
to construct some of the hardest queries.
Tuples referring to the same stock are correlated to one another.
For example, in \Cref{fig:portfolio}, tuples 1 and 2 are correlated to each other
but are independent of the other tuples.

\smallskip\noindent\emph{Data integration}:
The \tpch dataset consists of about 117,600 tuples extracted from the \tpch benchmark~\cite{tpch}.
We simulate the result of hypothetically integrating several data sources to form this data set:
we model uncertainty in each  attribute's value with discrete probability distributions.
For each original (deterministic) value in the \tpch dataset, we generate $D$ possible variations thereof,
where $D$ is the number of data sources that have been integrated into one.
The mean of these $D$ values is anchored around the original value;
each source value is sampled from an exponential, Poisson, uniform or Student's t-distribution.

\begin{figure*}
    \begin{mdframed}[linecolor=gray,innerleftmargin=10,innerrightmargin=10,innertopmargin=1,innerbottommargin=1]
    \centering{\scriptsize{
    \begin{tabularx}{\textwidth}{XcXcX}
    &
    \textbf{\naive (feasibility rate)}
    \includegraphics[width=5pt]{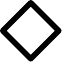} 0\%
    \includegraphics[width=5pt]{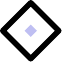} 25\%
    \includegraphics[width=5pt]{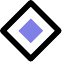} 50\%
    \includegraphics[width=5pt]{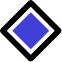} 75\%
    \includegraphics[width=5pt]{figs/plots/legend/mc_1_0.pdf} 100\%
    & &
    \textbf{\sss (feasibility rate)}
    \includegraphics[width=5pt]{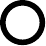} 0\%
    \includegraphics[width=5pt]{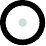} 25\%
    \includegraphics[width=5pt]{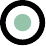} 50\%
    \includegraphics[width=5pt]{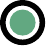} 75\%
    \includegraphics[width=5pt]{figs/plots/legend/ss_1_0.pdf} 100\%
    &
    \end{tabularx}}}
    \end{mdframed}
    \vspace*{2pt}
    \scriptsize{
    \centering
    \begin{tabularx}{1\textwidth}{m{0.2cm}m{0.1cm}XXXXXXXX}
    & &
    \hspace{1.3em}Q1 & \hspace{1.3em}Q2 & \hspace{1.3em}Q3 & \hspace{1.3em}Q4 &
    \hspace{1.3em}Q5 & \hspace{1.3em}Q6 & \hspace{1.3em}Q7 & \hspace{1.3em}Q8\\
    \multirow[c]{3}{*}[-0.2cm]{\rotatebox{90}{\textbf{Galaxy}}} &
    \rotatebox{90}{time (s)} &
    \plot{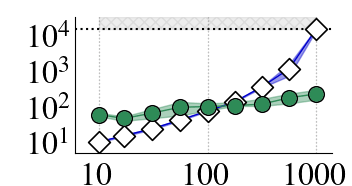}&
    \plot{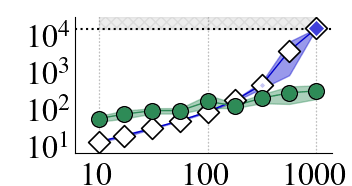}&
    \plot{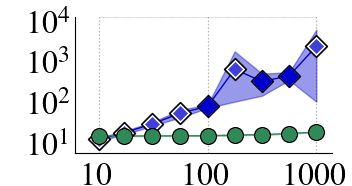}&
    \plot{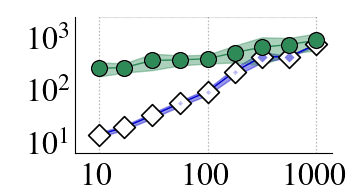}&
    \plot{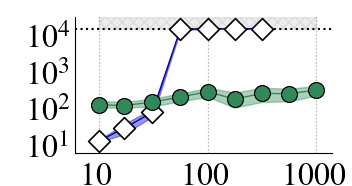}&
    \plot{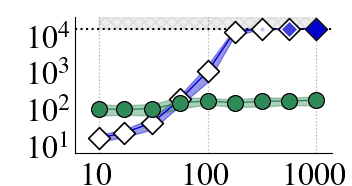}&
    \plot{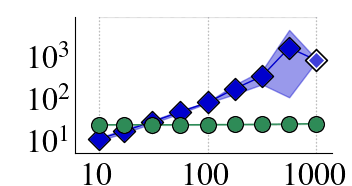}&
    \plot{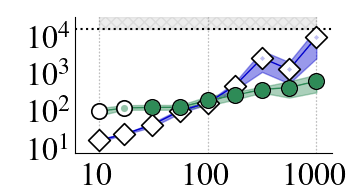}\\
    & \rotatebox{90}{$1 + \hat\epsilon$} &
    \plot{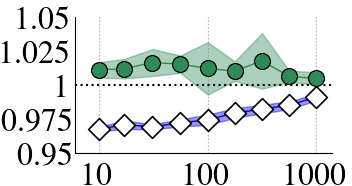}&
    \plot{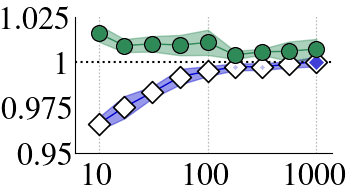}&
    \plot{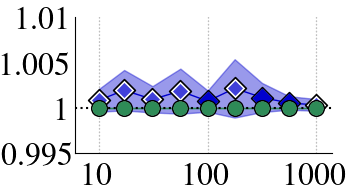}&
    \plot{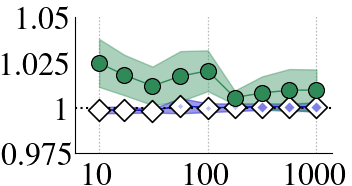}&
    \plot{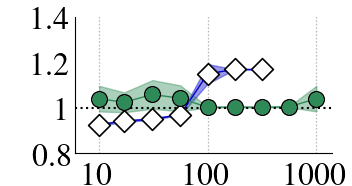}&
    \plot{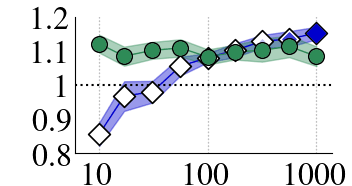}&
    \plot{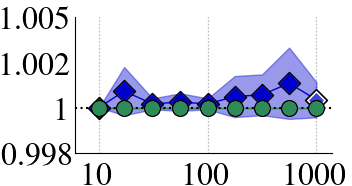}&
    \plot{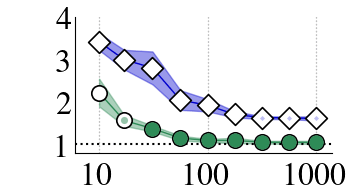}\\
    \end{tabularx}}
    {\color{gray}\noindent\rule{\textwidth}{0.1pt}}
    \centering\scriptsize{
    \begin{tabularx}{1\textwidth}{m{0.2cm}m{0.1cm}XXXXXXXX}
    \multirow[c]{3}{*}[-0.2cm]{\rotatebox{90}{\textbf{Portfolio}}} &
    \rotatebox{90}{\;\;time (s)} &
    \plot{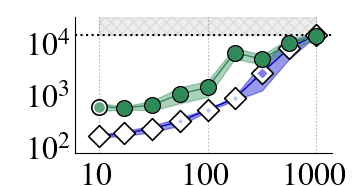}&
    \plot{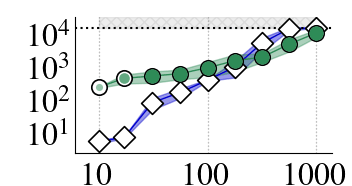}&
    \plot{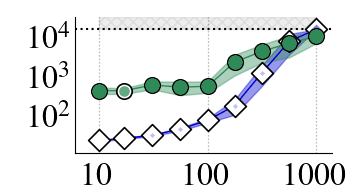}&
    \plot{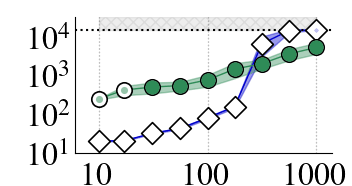}&
    \plot{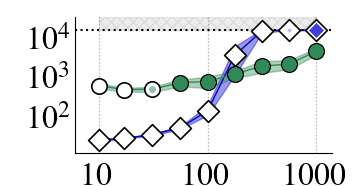}&
    \plot{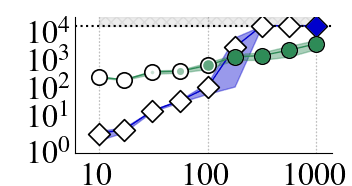}&
    \plot{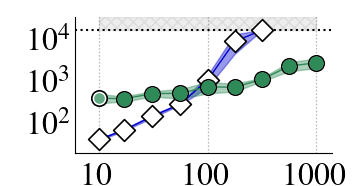}&
    \plot{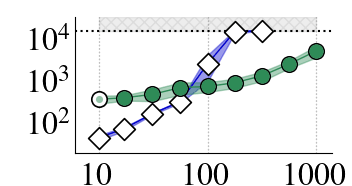}\\
    &
    \rotatebox{90}{$1+\hat\epsilon$}&
    \plot{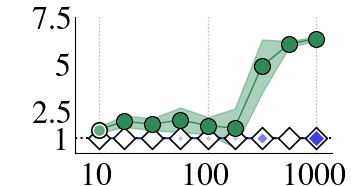}&
    \plot{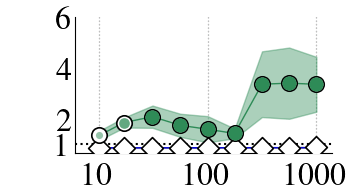}&
    \plot{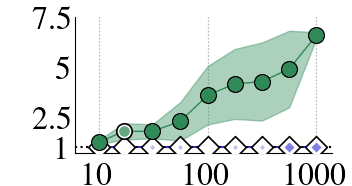}&
    \plot{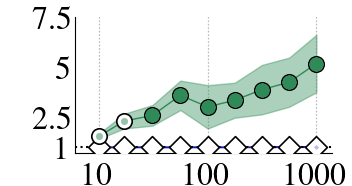}&
    \plot{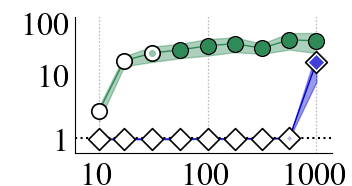}&
    \plot{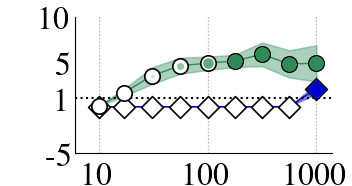}&
    \plot{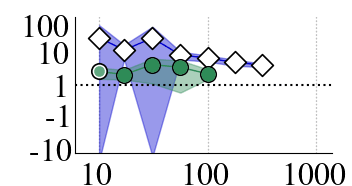}&
    \plot{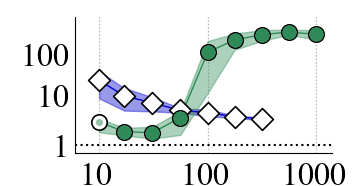}\\
    \end{tabularx}}
    {\color{gray}\noindent\rule{\textwidth}{0.1pt}}
    \centering\scriptsize{
    \begin{tabularx}{1\textwidth}{m{0.2cm}m{0.1cm}XXXXXXXX}
    \multirow[c]{3}{*}[-0.2cm]{\rotatebox{90}{\textbf{\tpch}}} &
    \rotatebox{90}{time (s)} &
    \plot{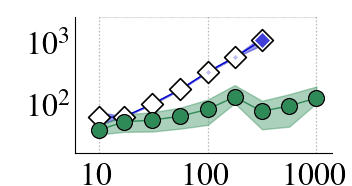}&
    \plot{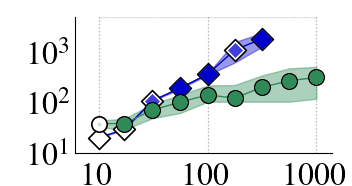}&
    \plot{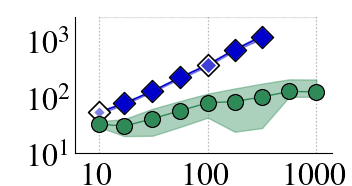}&
    \plot{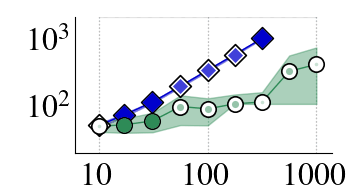}&
    \plot{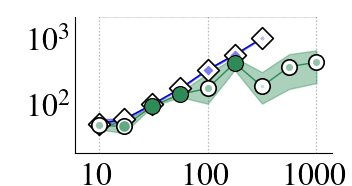}&
    \plot{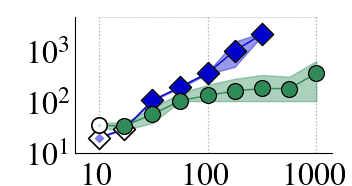}&
    \plot{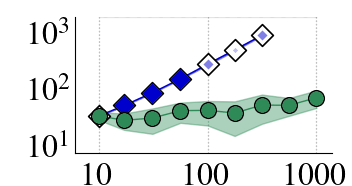}&
    \plot{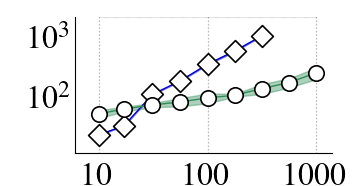}\\
    &
    \rotatebox{90}{$1+\hat\epsilon$}&
    \plot{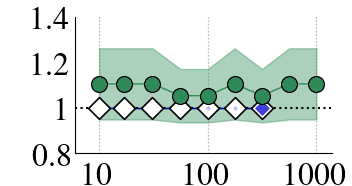}&
    \plot{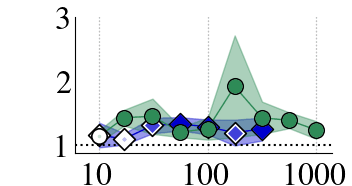}&
    \plot{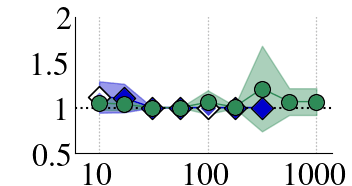}&
    \plot{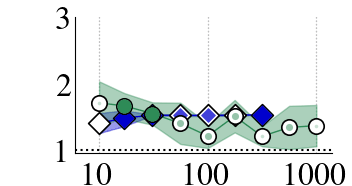}&
    \plot{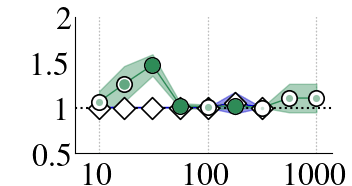}&
    \plot{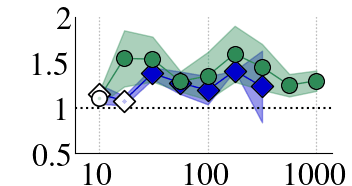}&
    \plot{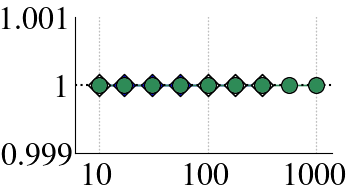}&
    \plot{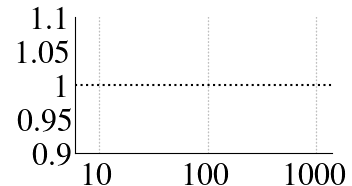}\\
    & &
    \hspace{1.5em}\textbf{N. Scenarios} & \hspace{1.5em}\textbf{N. Scenarios} & \hspace{1.5em}\textbf{N. Scenarios} &
    \hspace{1.5em}\textbf{N. Scenarios} & \hspace{1.5em}\textbf{N. Scenarios} & \hspace{1.5em}\textbf{N. Scenarios} &
    \hspace{1.5em}\textbf{N. Scenarios} & \hspace{1.5em}\textbf{N. Scenarios}
    \end{tabularx}}
    \caption{
    Scalability of \naive and \sss with increasing number of optimization scenarios.
    \naive struggles to find feasible solutions even with a large number of scenarios and often fails completely
    (missing points in the plot).
    \sss quickly finds feasible solutions with few scenarios.
    The approximation ratios of \sss's solutions are generally low when
    the number of scenarios is small.
    }
    \label{fig:n-scenarios-scale}
\end{figure*}

\smallskip
For each of the three datasets, we constructed a workload of eight \spaql queries;
all 24 queries, except one in \tpch, are feasible.
The workloads span seven different distributions for the uncertain data attributes,
including a complex VG function to predict future stock prices.
The objective functions are supported by the constraints for the Portfolio queries,
independent for the TPC-H queries and either supported or counteracted for the Galaxy queries
(see \Cref{def:supportiveness} for supported/counteracted/independent objectives).
The Portfolio workload tests high- and low-risk, high- and low-VaR (Value at Risk)---i.e., $p$ and $v$ in \Cref{eq:probconst}---as well as short- and long-term trade predictions.
The \tpch workload is split into queries with $D=3$ and $D=10$ (number of integrated sources).
For all queries there are two constraints, one of which is probabilistic with $p \geq 0.9$.
Examples include:
(1)~for Galaxy, we seek a set of five to ten sky regions that minimizes total expected radiation flux
while avoiding total flux levels higher than 40 with high probability, and
(2)~for \tpch, we seek a set of between one and ten transactions having maximum expected total revenue,
while containing less than 15 items total with high probability.
A detailed description of the workloads can be found in \Cref{sec:workload}.

\para{Evaluation metrics}
We measure \emph{response time} (in seconds and logarithmic scale)
across $10$ i.i.d runs using different seeds
for generating the optimization scenarios,
and evaluate feasibility and the objective value on an out-of-sample validation set
with $10^6$ scenarios ($10^7$ for the Portfolio workload).
We plot the average across the $10$ runs, and its $95\%$ confidence interval in a shaded area.
For each run of an algorithm, we set a time limit of four hours.
When the time limit expires, we interrupt \cplex and get the best solution found by the solver until then.
We measure \emph{feasibility rate} as the fraction, out of the $10$ runs,
in which a method produces a feasible solution (including, for all methods, when the time limit expired).
Because the true optimal solution for any of the queries is unknown,
we measure \emph{accuracy} by  $1+\hat\epsilon$,
where $\hat\epsilon \coloneqq \omega / \omega^* - 1$ and $\omega^*$ is
the objective value of the best feasible solution found by any of the methods.

\subsection{Results and Discussion} \label{subsec:results}
We evaluate four fundamental aspects of our algorithms:
(1)~query response time to reach $100\%$ feasibility rate;
(2)~scalability with increasing number of scenarios ($M$);
(3)~scalability of \sss with increasing number of summaries ($Z$);
(4)~scalability with increasing dataset size ($N$).

\newcolumntype{Y}{>{\centering\arraybackslash}X}
\begin{figure*}
    \begin{mdframed}[linecolor=gray,innerleftmargin=10,innerrightmargin=10,innertopmargin=1,innerbottommargin=1]
    \centering{\scriptsize{\begin{tabularx}{\textwidth}{XcXcX}
    &
    \textbf{\naive (feasibility rate)}
    \includegraphics[width=5pt]{figs/plots/legend/mc_0_0.pdf} 0\%
    \includegraphics[width=5pt]{figs/plots/legend/mc_0_25.pdf} 25\%
    \includegraphics[width=5pt]{figs/plots/legend/mc_0_5.pdf} 50\%
    \includegraphics[width=5pt]{figs/plots/legend/mc_0_75.pdf} 75\%
    \includegraphics[width=5pt]{figs/plots/legend/mc_1_0.pdf} 100\%
    & &
    \textbf{\sss (feasibility rate)}
    \includegraphics[width=5pt]{figs/plots/legend/ss_0_0.pdf} 0\%
    \includegraphics[width=5pt]{figs/plots/legend/ss_0_25.pdf} 25\%
    \includegraphics[width=5pt]{figs/plots/legend/ss_0_5.pdf} 50\%
    \includegraphics[width=5pt]{figs/plots/legend/ss_0_75.pdf} 75\%
    \includegraphics[width=5pt]{figs/plots/legend/ss_1_0.pdf} 100\%
    &
    \end{tabularx}}}\end{mdframed}
    \vspace*{2pt}
\centering\scriptsize{
\begin{tabularx}{1\textwidth}{m{0.1cm}XXXXXXXX}
    &
    \hspace{1.3em}Q1 & \hspace{1.3em}Q2 & \hspace{1.3em}Q3 & \hspace{1.3em}Q4 &
    \hspace{1.3em}Q5 & \hspace{1.3em}Q6 & \hspace{1.3em}Q7 & \hspace{1.3em}Q8\\
    \rotatebox{90}{time (s)}&
    \plot{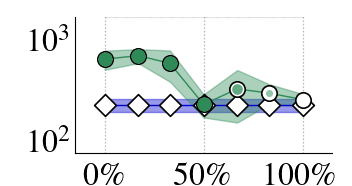}&
    \plot{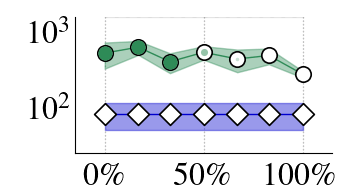}&
    \plot{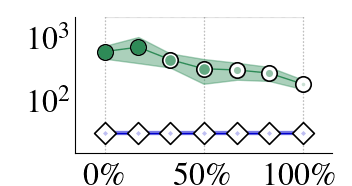}&
    \plot{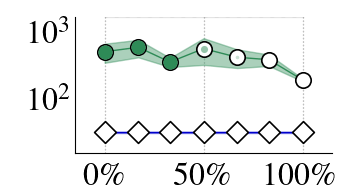}&
    \plot{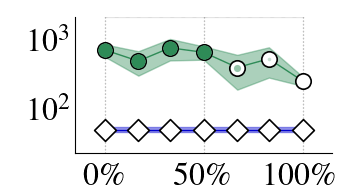}&
    \plot{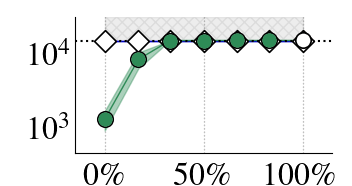}&
    \plot{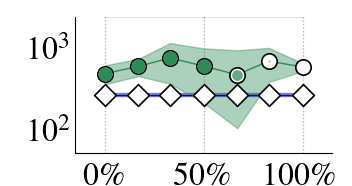}&
    \plot{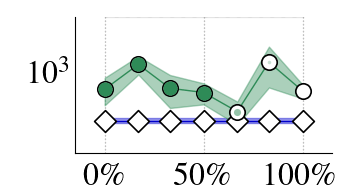}\\
    \rotatebox{90}{$1+\hat\epsilon$}&
    \plot{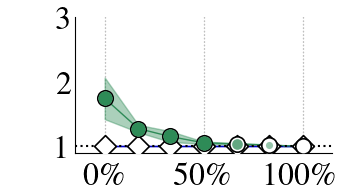}&
    \plot{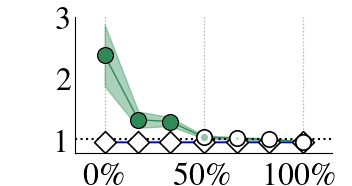}&
    \plot{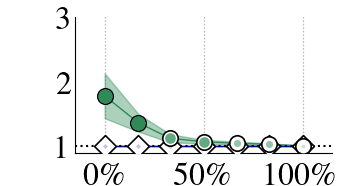}&
    \plot{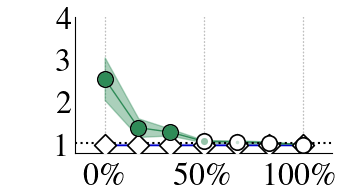}&
    \plot{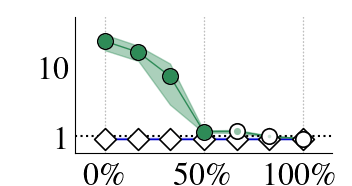}&
    \plot{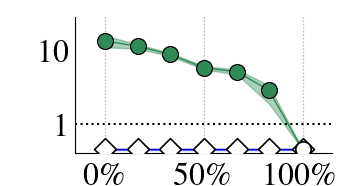}&
    \plot{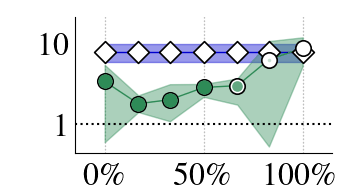}&
    \plot{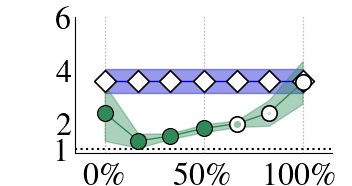}\\
    & \hspace{1.2em}\textbf{\% N. Summaries} & \hspace{1.2em}\textbf{\% N. Summaries}
    & \hspace{1.2em}\textbf{\% N. Summaries} & \hspace{1.2em}\textbf{\% N. Summaries}
    & \hspace{1.2em}\textbf{\% N. Summaries} & \hspace{1.2em}\textbf{\% N. Summaries}
    & \hspace{1.2em}\textbf{\% N. Summaries} & \hspace{1.2em}\textbf{\% N. Summaries}
\end{tabularx}}
    \caption{
    Effects of increasing number of summaries ($Z$) on the Portfolio workload,
    as a percentage of the number of scenarios,
    from 1 summary up to $M$ summaries ($100\%$).
    Increasing $Z$ improves the approximation ratio of the solution produced by \sss.
    Increasing $Z$ too far results in infeasible solutions as, when $Z=M$,
    \sss is identical to \naive,
    and it thus overfits, like \naive, to a bad set of scenarios.
    }
    \label{fig:n-summaries-scale}
\end{figure*}

\subsubsection{Response time to reach $100\%$ feasibility rate}
Both \naive and \sss increase $M$ (the number of scenarios) up to when solutions start to be feasible.
We report the cumulative time for all iterations the algorithm took to reach a certain feasibility rate,
from $0\%$, up to $100\%$ (when the algorithm produces feasible solutions for all $10$ runs).
For \sss, $Z$ is fixed ($1$ for Galaxy and Portfolio, $2$ for \tpch).
We set $Z$ to the lowest value (per workload) such that \sss could reach 100\% feasibility rate.
\Cref{fig:time-to-feasible} shows the results of the experiment.
For all (23) feasible queries across all workloads, \sss is always able to reach $100\%$ feasibility rate,
while \naive can only reach $100\%$ feasibility for only $7$ queries.
Even then, \sss is usually orders of magnitude faster than \naive
(e.g., Galaxy Q6, \tpch Q2, Q6, and Q7).
Moreover, in 15 out of the 23 feasible queries, \sss reached 100\% feasibility while \naive was still at $0\%$.
The conservative nature of summaries allows higher feasibility rates for \sss
even with fewer scenarios. As the number of scenarios increases,
\sss solves a much smaller problem than \naive, leading to orders-of-magnitude faster response time.

The only case where \sss is slower than \naive at reaching $100\%$ feasibility rate is Galaxy Q7,
which was an easy query for both methods:
both solved it with only 10 scenarios.
This query has a supported objective function over data with minimal uncertainty
described by a Pareto distribution with ``scale'' and ``shape'' both equal to $1$.
For this query, the summarization process and solving a probabilistically-unconstrained problem are overheads for \sss.
\tpch Q8 is an infeasible query.
Both methods increase $M$ up to $1000$ before declaring infeasibility,
but again \sss is faster than \naive in doing so.

\vspace{3pt}
\subsubsection{Effect of increasing the number of optimization scenarios}
We evaluate the scalability of our methods
when the number of optimization scenarios $M$ increases; $Z$ is fixed as described above.
For each algorithm, we group feasibility rates into 5 groups: 0\%, 25\%, 50\%, 75\% and 100\%,
and use different shadings to distinguish each case.

\Cref{fig:n-scenarios-scale} gives scalability results for the three workloads.
Generally, with low $M$, \naive executes very quickly to produce infeasible solutions with low objective values
(optimizer's curse); as \naive increases $M$, the running time increases exponentially---note the logarithmic scale---up to a point where it fails altogether (missing \naive points in the plots).
On the other hand, \sss finds feasible solutions even with as little as 10 scenarios.

\sss produces high  quality solutions
as demonstrated by the low approximation ratio ($1+\hat\epsilon$), close to $1$ for most queries.
However, with the hardest Portfolio queries (Q5 and Q6),
the worst approximation ratio for \sss is quite high for feasible solutions:
this is an indicator that the number of summaries, $Z=1$ is too low and should be increased.

\newcolumntype{Y}{>{\centering\arraybackslash}X}
\begin{figure*}
    \begin{mdframed}[linecolor=gray,innerleftmargin=10,innerrightmargin=10,innertopmargin=1,innerbottommargin=1]
    \centering{\scriptsize{\begin{tabularx}{\textwidth}{XcXcX}
    &
    \textbf{\naive (feasibility rate)}
    \includegraphics[width=5pt]{figs/plots/legend/mc_0_0.pdf} 0\%
    \includegraphics[width=5pt]{figs/plots/legend/mc_0_25.pdf} 25\%
    \includegraphics[width=5pt]{figs/plots/legend/mc_0_5.pdf} 50\%
    \includegraphics[width=5pt]{figs/plots/legend/mc_0_75.pdf} 75\%
    \includegraphics[width=5pt]{figs/plots/legend/mc_1_0.pdf} 100\%
    & &
    \textbf{\sss (feasibility rate)}
    \includegraphics[width=5pt]{figs/plots/legend/ss_0_0.pdf} 0\%
    \includegraphics[width=5pt]{figs/plots/legend/ss_0_25.pdf} 25\%
    \includegraphics[width=5pt]{figs/plots/legend/ss_0_5.pdf} 50\%
    \includegraphics[width=5pt]{figs/plots/legend/ss_0_75.pdf} 75\%
    \includegraphics[width=5pt]{figs/plots/legend/ss_1_0.pdf} 100\%
    &
    \end{tabularx}}}\end{mdframed}
    \vspace*{2pt}
\centering\scriptsize{
\begin{tabularx}{1\textwidth}{m{0.1cm}XXXXXXXX}
    &
    \hspace{1.3em}Q1 & \hspace{1.3em}Q2 & \hspace{1.3em}Q3 & \hspace{1.3em}Q4 &
    \hspace{1.3em}Q5 & \hspace{1.3em}Q6 & \hspace{1.3em}Q7 & \hspace{1.3em}Q8\\
    \rotatebox{90}{time (s)}&
    \plot{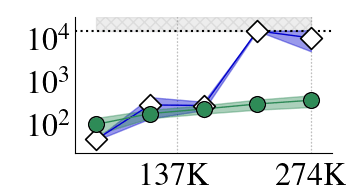}&
    \plot{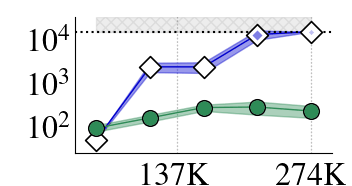}&
    \plot{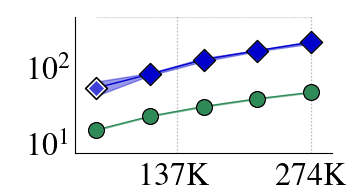}&
    \plot{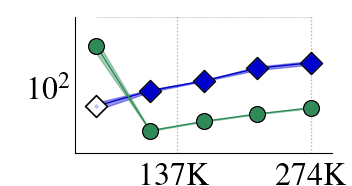}&
    \plot{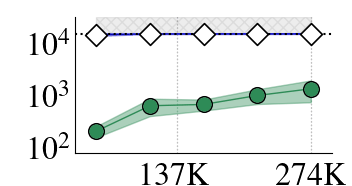}&
    \plot{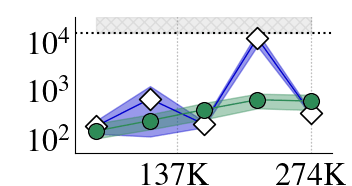}&
    \plot{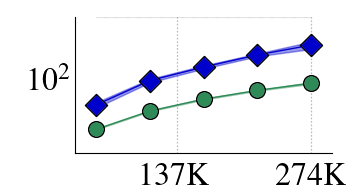}&
    \plot{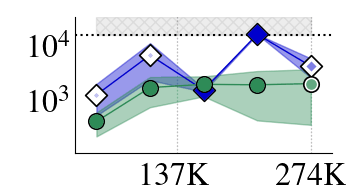}\\
    \rotatebox{90}{$1+\hat\epsilon$}&
    \plot{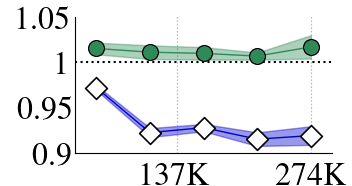}&
    \plot{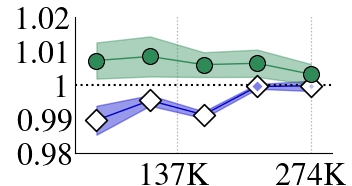}&
    \plot{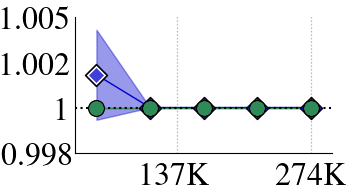}&
    \plot{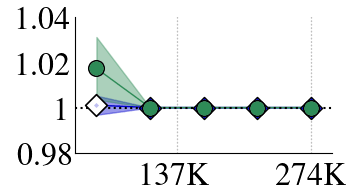}&
    \plot{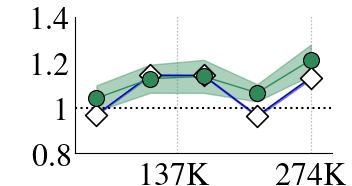}&
    \plot{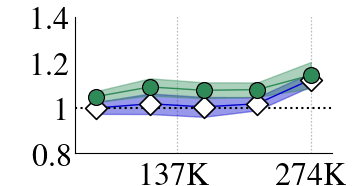}&
    \plot{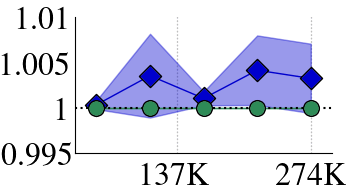}&
    \plot{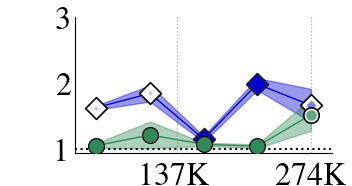}\\
    &
    \hspace{1.5em}\textbf{N. Tuples} &
    \hspace{1.5em}\textbf{N. Tuples} &
    \hspace{1.5em}\textbf{N. Tuples} &
    \hspace{1.5em}\textbf{N. Tuples} &
    \hspace{1.5em}\textbf{N. Tuples} &
    \hspace{1.5em}\textbf{N. Tuples} &
    \hspace{1.5em}\textbf{N. Tuples} &
    \hspace{1.5em}\textbf{N. Tuples}
\end{tabularx}}
    \vspace{-3mm}
    \caption{
    Scalability of \naive and \sss with increasing dataset size ($N$) on the Galaxy workload.
    The running times of both algorithms degrades with increasing $N$, but \sss scales up well in comparison with \naive.
    }
    \label{fig:data-size-scale}
\end{figure*}

\subsubsection{Effect of increasing the number of summaries} \label{subsec:increasing-summaries}
In this experiment, we show how increasing the number of summaries ($Z$) helps improve the approximation ratio
in the Portfolio queries.
We increase $Z$ from $1$ up to $M$ (number of scenarios),
where $M$ is set to where the feasibility rate of \sss was $100\%$ in the previous experiment,
and we show the running time and approximation ratio
compared to \naive with $M$ scenarios.
Figure~\ref{fig:n-summaries-scale} shows the results of this experiment.
First, the response time with increasing $Z$ is in most cases independent of $Z$.
In fact, while increasing $Z$ adds more scenarios to the \CSA formulation,
each summary becomes less and less conservative, making the problem a bit larger but always easier;
in the limit ($Z=M$), each summary is identical to an original scenario,
and thus \sss only pays the extra overhead, compared to \naive,
of solving the probabilistically-unconstrained problem first.
On the other hand, \naive is always faster, but its solutions are infeasible.
For most queries,
the approximation ratio closely approaches $1$,
while still maintaining a high feasibility rate.
Increasing $Z$ too far eventually causes
feasibility to drop,
reaching that of \naive in the limit ($Z=M$).

Finally, even though infeasible solutions tend to have better objective values than feasible ones,
we find that \naive's infeasible solutions to Q7 and Q8 have worse objective values.
These queries proved quite challenging for \naive as they involved stock price predictions for a week in the future.

\subsubsection{Effect of increasing the dataset size}
In this experiment, we increase the Galaxy dataset up to five times from 55,000 tuples to 274,000 tuples.
For all queries except Q8 we fix $M=56$ (for both algorithms) and $Z=1$.
In general, \sss scales well with increasing data set size: it finds feasible solutions
with good approximation ratios.
\naive, however, times out for several queries (Q1, Q2, Q5, Q6, \& Q8)
and its response time sharply increases as dataset size increases (Q1, Q2, Q6, Q8).
Except for three queries (Q3, Q4, Q7), most of \naive's solutions are infeasible;
even then, \sss produces feasible solutions in orders of magnitude less time with better approximation ratios.

In Q8, we set $M=562$ to enable \sss to still produce feasible solutions
(75\% feasibility at 274K tuples), without causing \naive to fail.
Q8 is a challenging query as each data value is sampled from a Pareto distribution with different parameters
leading to high variability across scenarios.

To further increase the data size scalability of \sss,
we hope to combine it with partitioning and divide-and-conquer approaches similar to \sr~\cite{Brucato2018}.

\section{Related Work} \label{sec:related-work}

\para{Probabilistic databases and package queries}
\emph{Probabilistic databases}~\cite{dalvi2007efficient,suciu2011probabilistic}
have focused mainly on modeling discrete data uncertainty;
the \emph{Monte Carlo Database} (MCDB)~\cite{jampani2008mcdb}
supports arbitrary uncertainty, via VG functions.
Probabilistic databases support \sql queries, but lack support for optimization.
\emph{Package query engines}~\cite{Brucato2018,vsikvsnys2016solvedb}
offer support only  for deterministic optimization.

\para{Stochastic optimization}
\emph{Stochastic optimization}~\cite{HOMEMDEMELLO201456} studies
approximations for stochastic constraints and objectives.
Probabilistic constraints are very hard to handle in general, because the feasible region of the inner constraint
may be non-convex~\cite{HOMEMDEMELLO201456,ahmed2008solving,calafiore2006probabilistic,CAMPI2009149,nemirovski2006scenario,dentcheva2006optimization,Luedtke2010}.
In this work, we study stochastic optimization problems with objective functions and constraints
defined in terms of linear functions of the tuple attributes.

Our \naive method
is derived from the numerous ``scenario approximations'' from the \stpr
literature~\cite{HOMEMDEMELLO201456,kall1994stochastic,calafiore2006probabilistic,CAMPI2009149,nemirovski2006scenario,luedtke2008sample,campi2011sampling,nemirovski2006convex}.
Choosing the number of scenarios ($M$) a priori is one of the most studied problems.
Campi et al.~\cite{CAMPI2009149} show that
the optimal solution of a Monte Carlo formulation that satisfies exactly $M$ i.i.d. scenarios
is feasible with probability at least $\delta$
if $M \ge \frac{2}{1-p_j}\left(\ln\left(\frac{1}{1-\delta}\right) + N\right)$.
A-priori bounds quickly become impractical in a database setting,
where $N$ is also the number of tuples, and thus typically large.
For example, with a table of size $N=50{,}000$, $p_j=0.9$, $\delta=0.95$, at least $M \ge 1{,}000{,}060$
scenarios must be generated and all satisfied.

\emph{Scenario removal} studies techniques for removing scenarios after
sampling~\cite{campi2011sampling,dupavcova2003scenario,karuppiah2010simple,luedtke2008sample,Calafiore2005}.
Empirically, these methods generally provide a reduction factor of only $50\%$ or less,
which is insufficient for our setting.
Our $\alpha$-summary can be viewed as removing $100 (1-\alpha) \%$ of the scenarios, where $\alpha$ is usually very small (below $0.01$);
not only do we remove scenarios, but we replace them with conservative summaries.

Similar to our setting, \emph{distributionally robust optimization}
(DRO)~\cite{hanasusanto2015distributionally,delage2010distributionally,lam2018sampling}
attempts to mitigate the optimizer's curse
when the uncertainty distribution is unknown but is assumed to lie in some set of candidate distributions;
the original probability constraints are replaced with worst-case probability constraints based on this set.
In contrast,  \sss uses deterministic worst-case constraints, which are simpler and 
avoid assumptions on the uncertainty distribution. DRO methods also show limited scalability in the number of variables $N$, e.g., 
$N$ is at most 20 in the experiments in~\cite{lam2018sampling}.

The goal of \emph{wait-and-judge optimization}~\cite{campi2018wait,campi2018general}
is to perform a-posteriori feasibility analysis.
Existing approaches help provide bounds on the quality of a solution,
but do not provide algorithms that dynamically adapt in response to poor solutions.
\sss, instead, adjusts the conservativeness of the summaries to obtain feasible solutions
with minimum computational cost.
\sss can potentially use wait-and-judge during out-of-sample validation to 
 decide when to stop increasing the number of scenarios.

\vspace*{-2px}
\section{Conclusion and Future Work} \label{sec:conclusion-and-future}

In this paper, we addressed \emph{single-stage} decision making under uncertainty,
in which decisions are made before the values of the random variables become known.
In many cases, however, uncertainty is revealed over time, in stages, allowing for remedial actions.
We plan to explore these dynamic settings, referred to as stochastic programming with recourse.
Another goal is to extend our methods to problems that involve probabilistic constraints where the inner constraints
must \emph{jointly} be satisfied with a given probability;
such an extension is highly nonntrivial.
We also plan to work on further algorithmic improvements, including (i)
developing more sophisticated summarization methods than minimum and maximum summaries;
(ii) scaling up \sss to very large datasets (e.g., millions of tuples) by combining
summaries
with divide-and-conquer approaches like \sr~\cite{Brucato2018}; (iii)
parallelizing \CSAsolve
and  summary generation;
and (iv) fully integrating stochastic package queries into a probabilistic database to handle multi-table queries.
We plan to further develop our theory on \sss to formally prove its convergence to feasible solutions
as the number of scenarios increases.
Finally, we plan to explore ways to ``open the black box'' of optimization software
to allow for further performance improvements,
in analogy to the way MCDB re-engineered query operations to efficiently handle uncertain tuple attributes.

\vspace*{-2px}
\para{Acknowledgments}
This work was supported by the NYUAD Center for Interacting Urban Networks (CITIES),
and funded by:
Tamkeen under the NYUAD Research Institute Award CG001,
the Swiss Re Institute under the Quantum Cities initiative, and
the National Science Foundation under grants IIS-1453543 and IIS-1943971.
The authors would like to thank the anonymous reviewers for their valuable insights,
and Arya Mazumdar, Nishad Ranade, and Senay Solak for their help and suggestions
during various phases of the work.

\bibliographystyle{abbrv}
\bibliography{refs}

\vfill
\appendix

\section{\spaql: Language Support} \label{sec:spaql}
In this appendix, we provide additional details about our language \spaql.
Figure~\ref{fig:spaql-grammar} shows the syntax diagram for \spaql, constructed using
the Railroad Diagram Generator~\cite{railroad}.
\spaql extends \paql~\cite{Brucato2018}
to support expectation and probabilistic constraints and objectives in the following ways.
In \paql, a linear constraint has the general form:

\begin{sqlquerylg}
\> (\SELECT \SUM{$f(\rel{R})$} \WHERE <selection-predicate> \FROM $\p$) $\ge v$
\end{sqlquerylg}
where $\p$ is a reference name (alias) to the result package,
$f(\rel{R})$ a function of the attributes of $\rel{R}$
(e.g., $f(\rel{R}) = 3\attr{A}^2_1 - 2\sqrt{\attr{A}_2} + 1$ for a table $\rel{R}$ with
two attributes $\attr{A}_1$ and $\attr{A}_2$),
and $v \in \reals$.
Syntactic sugar for a simple single-attribute, no-selection constraint is
$\SUM{\attr{A}} \ge v$, where $f(\rel{R}) = \attr{A}$, for some attribute $\attr{A}$.
For example, $\SUM{\dattr{price}} \le 1000$ from the query in the introduction
is a single-attribute summation constraint on $\dattr{price}$.
A cardinality constraint is a special case of a summation constraint,
$\COUNT{\ast} = \SUM{1}$, where $f(\rel{R}) = 1$.

If any of the attributes in $f(\rel{R})$ are stochastic, \spaql allows users to write either
an expected or a probabilistic version of the constraint.
An expected constraint simply prepends the keyword $\EXPECTED{}$ to a deterministic constraint, e.g.,
$\EXPECTED{\SUM{\attr{A}}} \ge v$.
Similarly, an expected minimization objective can be expressed as
$\MINIMIZE\; \EXPECTED{\SUM{\attr{A}}}$.
For example, the objective function of query $\query$ from the introduction
maximizes the $\EXPECTED{\SUM{\attr{Gain}}}$.

A probabilistic constraint can be expressed by appending $\WITHPROBABILITY{\ge p}$ to a deterministic constraint,
for some $p \in (0,1)$.
For example, $\SUM{\attr{Gain}} \ge -10$ $\WITHPROBABILITY{\ge 0.95}$.
The language also allows for opposite constraints ($\le p$) for convenience,
but, as noted earlier, they can always be equivalently rewritten in the other form.
A probabilistic objective is expressed by prepending $\PROBABILITYOF{}$ to a constraint,
e.g., $\MAXIMIZE\; \PROBABILITYOF{\SUM{\attr{Gain}} \ge -10}$.

\begin{figure*}
\begin{tabular}{lll}
    \multicolumn{3}{c}{
        \makecell[tl]{
            \textbf{PackageQuery}:\\
            \includegraphics[scale=0.58]{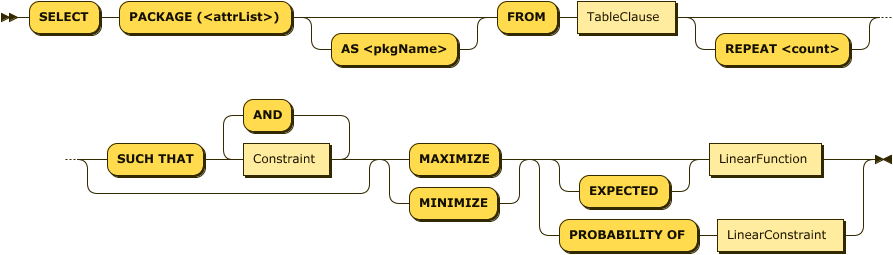}
        }
    }\\
     &\;\;\;& \\
    \makecell[tl]{
        \textbf{TableClause}:\\
        \includegraphics[scale=0.58]{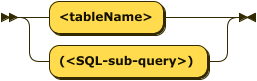}
    } & &
    \makecell[tl]{
        \textbf{LinearFunction}:\\
        \includegraphics[scale=0.58]{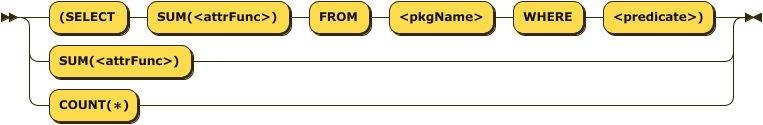}
    }\\
     &\;\;\;& \\
    \makecell[tl]{
        \textbf{LinearConstraint}:\\
        \includegraphics[scale=0.58]{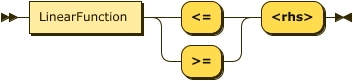}
    } & &
    \makecell[tl]{
        \textbf{Constraint}:\\
        \includegraphics[scale=0.58]{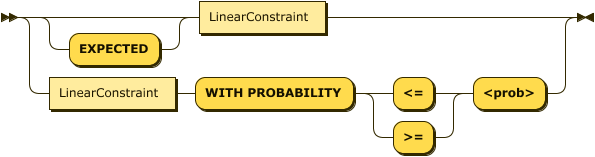}
    }\\
\end{tabular}
\caption{Syntax (railroad) diagram of \spaql.}
\label{fig:spaql-grammar}
\end{figure*}

\section{Approximation Guarantees} \label{sec:optimality}
\setlength{\tabcolsep}{8pt}

In this section, we provide more details about
our theoretical approximation guarantees (\Cref{subsec:early-termination}),
including proofs that were omitted in the main body of the paper,
and all the remaining cases that were skipped for space constraints.

\subsection{Objective types and signs}\label{subsec:objective-types-and-signs}

Recall that in \Cref{th:unsupp-approx-mc} we assumed
a minimization query with nonnegative objective values.
The following propositions replace \Cref{th:unsupp-approx-mc}
under different conditions.

\para{Minimization with negative objective values}
\begin{proposition} \label{th:min-neg}
    Let $\epsilon \ge 0$ and let $\underline{\omega}$
    be a negative constant such that $\underline{\omega} \le \hat\omega$.
    Set $\epsilon^{(q)}=(\underline{\omega}/\omega^{(q)})-1$.
    If $\epsilon^{(q)} \le \epsilon$,
    then $\hat\omega \ge {(1 + \epsilon)} \omega^{(q)}$.
\end{proposition}
\begin{proof}
    Suppose that $\epsilon^{(q)} \le \epsilon$.
    We have
    \[
        \hat\omega
        \ge \underline{\omega}
        =   \biggl(1+\Bigl(\frac{\underline{\omega}}{\omega^{(q)}}-1\Bigr)\biggr)\omega^{(q)}
        =   \bigl(1+\epsilon^{(q)}\bigr)\omega^{(q)}
        \ge (1+\epsilon)\omega^{(q)},
    \]
    and the result follows.
\end{proof}

\para{Maximization with nonnegative objective values}
\begin{proposition} \label{th:max-nonneg}
    Let $\epsilon \ge 0$ and let $\overline{\omega}$
    be a positive constant such that $\hat\omega \le \overline{\omega}$.
    Set $\epsilon^{(q)}=(\overline{\omega}/\omega^{(q)})-1$.
    If $\epsilon^{(q)} \le \epsilon$,
    then $\hat\omega \le {(1 + \epsilon)} \omega^{(q)}$.
\end{proposition}
\begin{proof}
    Suppose that $\epsilon^{(q)} \le \epsilon$.
    We have
    \[
        \hat\omega
        \le \overline{\omega}
        =   \biggl(1+\Bigl(\frac{\overline{\omega}}{\omega^{(q)}}-1\Bigr)\biggr)\omega^{(q)}
        =   \bigl(1+\epsilon^{(q)}\bigr)\omega^{(q)}
        \le (1+\epsilon)\omega^{(q)},
    \]
    and the result follows.
\end{proof}

\para{Maximization with negative objective values}
\begin{proposition} \label{th:max-neg}
    Let $\epsilon \ge 0$ and let $\overline{\omega}$
    be a negative constant such that $\hat\omega \le \overline{\omega}$.
    Set $\epsilon^{(q)}=(\omega^{(q)} / \overline{\omega})-1$.
    If $\epsilon^{(q)} \le \epsilon$,
    then $\omega^{(q)} \ge {(1 + \epsilon)}\hat\omega$.
\end{proposition}
\begin{proof}
    Suppose that $\epsilon^{(q)} \le \epsilon$.
    Since $\hat\omega/\overline{\omega} \ge 1$ and $\omega^{(q)} < 0$, we have
    \[
        \omega^{(q)}
        \ge \Bigl(\frac{\hat\omega}{\overline{\omega}}\Bigr) \omega^{(q)}
        =   \biggl(1+\Bigl(\frac{\omega^{(q)}}{\overline{\omega}}-1\Bigr)\biggr)\hat\omega
        =   \bigl(1+\epsilon^{(q)}\bigr)\hat\omega
        \ge (1+\epsilon)\hat\omega,
    \]
    and the result follows.
\end{proof}

\subsection{Upper and lower bounds on $\hat\omega$}\label{subsec:upper-and-lower-bounds-tonull}
Our theory uses upper and lower bounds on the optimal objective value $\hat\omega$
to derive approximation bounds.
We provide bounds under the following assumptions:
(A1)~there exist bounds on the values of the validation scenarios,
(A2)~there exist package size bounds.
These two assumptions are not too restrictive since we can almost always find such bounds
by analyzing the query, or the validation scenarios produced by the VG functions.
The following are examples of simple bounds for (A1) and (A2):

\para{(A1) Validation scenarios bounds}
We assume the availability of upper and lower bounds on the values of
the validation scenarios across the tuples in the optimal package.
That is, there should exist $\underline{s}$ and $\overline{s}$
such that ${\underline{s} \le \hat{s}_{ij}.\attr{A} \le \overline{s}}$,
for all $i \in [1..N] : \hat{x}_i > 0$, and $j \in [1..\hat{M}]$.
We can easily derive (possibly loose) bounds, by taking the minimum and maximum
scenario values across all input tuples,
i.e., by setting
$\underline{s} \coloneqq \min\{ \hat{s}_{ij}.\attr{A} \mid i \in [1..N], j \in [1..\hat{M}]\}$, and
$\overline{s} \coloneqq \max\{ \hat{s}_{ij}.\attr{A} \mid i \in [1..N], j \in [1..\hat{M}]\}$.
In principle, tighter bounds might exist.
For example, if we could identify tuples that cannot be part of the optimal solution,
we could take them out of the $\min$ and $\max$ in the above formulas.
In this work, we do not explore ways to derive better bounds.

\para{(A2) Package size bounds}
We also assume there exist upper and lower bounds on the size of
the optimal package.
That is, there exist $\underline{l}$ and $\overline{l}$ such that
$\underline{l} \le \sum_{i=1}^{N} \hat{x}_i \le \overline{l}$.
An obvious value for $\underline{l}$, always true, is $\underline{l} = 0$.
If the package query includes a cardinality constraint (i.e., a constraint on the $\COUNT{\ast}$),
this might be used directly to derive $\underline{l}$, $\overline{l}$, or both.
If the query includes deterministic summation constraints
(i.e., on $\SUM{\attr{A}}$, for a deterministic attribute $\attr{A}$),
we can derive bounds following the derivations presented in~\cite{brucato2014packagebuilder}.
Again, in this work we do not study ways to derive tighter bounds than the obvious ones.

\smallskip
We provide two types of bounds on $\hat\omega$:
(B1)~constraint-agnostic bounds, which are always available regardless of the probabilistic constraints;
(B2)~constraint-specific bounds, which depend on the probabilistic constraint and on whether the constraint
supports or counteracts the objective function, or it is independent of it (see~\Cref{def:supportiveness}).
In cases where both (B1) and (B2) are available, the final bound is the best of the two.

\para{(B1) Constraint-agnostic bounds}
In \Cref{tab:generic-bounds}, we show bounds on the optimal objective value of the form
$\underline\omega \le \hat\omega \le \overline\omega$.

\subsubsection*{Proofs}
The proof for $\underline\omega$ for case $\underline{s} \ge 0$ was provided in \Cref{subsec:early-termination}.
Similar derivations can be used for $\underline\omega$ under $\underline{s} < 0$ and for $\overline\omega$.
For example, for $\underline\omega$ under $\underline{s} < 0$, we have:
\[
    \hat\omega =\frac{1}{\hat M}\sum_{j=1}^{\hat M}\sum_{i=1}^N\hat{s}_{ij}.\attr{A}\;\hat{x}_i
    \ge \frac{1}{\hat M}\sum_{j=1}^{\hat M}\underline{s}\overline{l}
    = \underline{s}\overline{l},
\]
using the fact that $\sum_{i=1}^N \hat{x}_i \le \overline{l}$.
The other derivations follow a similar reasoning.

\para{(B2) Constraint-specific bounds}
Another class of bounds exist for an objective function that is supported or counteracted by
at least one probabilistic constraint of the form
$\probbg{\sum_{i=1}^{N} \xi_i x_i \odot v} \ge p$.
We first set the followings:
\begin{align*}
\hat{S}^\odot_{\hat{x}} &\coloneqq \{ j \in [1..\hat{M}] \mid \textstyle\sum_{i=1}^{N} \hat{s}_{ij}.\attr{A}\;\hat{x}_i \odot v \},\\
\hat{S}^\otimes_{\hat{x}} &\coloneqq [1..\hat{M}] \setminus \hat{S}^\odot_{\hat{x}},\\
\hat\omega^\odot &\coloneqq \frac{1}{\hat M} \sum_{j \in \hat{S}^\odot_{\hat{x}}} \sum_{i=1}^{N} \hat{s}_{ij}.\attr{A}\;\hat{x}_i,\\
\hat\omega^\otimes &\coloneqq \frac{1}{\hat M} \sum_{j \in \hat{S}^\otimes_{\hat{x}}} \sum_{i=1}^{N} \hat{s}_{ij}.\attr{A}\;\hat{x}_i.
\end{align*}
Intuitively, $\hat{S}^\odot_{\hat{x}}$ is the set of validation scenarios satisfied by the optimal solution $\hat{x}$,
and $\hat{S}^\otimes_{\hat{x}}$ is the set of validation scenarios \emph{not satisfied} by $\hat{x}$.
Notice that the optimal value is $\hat\omega = \hat\omega^\odot + \hat\omega^\otimes$.

We provide bounds in the form
$\underline\omega^\odot + \underline\omega^\otimes \le \hat\omega \le
     \overline\omega^\odot + \overline\omega^\otimes$.
They are implied by the following conditions:
\begin{align*}
\text{(C1)~} \underline\omega^\odot \le \hat\omega^\odot \text{ \;\;(C2)~ } \underline\omega^\otimes \le \hat\omega^\otimes,\\
\text{(C3)~} \hat\omega^\odot \le \overline\omega^\odot \text{ \;\;(C4)~ } \hat\omega^\otimes \le \overline\omega^\otimes.
\end{align*}

\Cref{tab:dependent-objective-bounds} shows all the available bounds for (C1-4) under different cases.
Recall, from \Cref{subsec:early-termination}, that our algorithm
uses the best available bounds (i.e., maximum lower bound, and minimum upper bound)
among all the available ones,
including the ones in \Cref{tab:generic-bounds}.

\begingroup
\renewcommand{\arraystretch}{1.2}
\begin{table}
\centering
    \begin{tabular}{l|ll}
        \toprule

        \textbf{Case} &
        $\underline\omega$ &
        $\overline\omega$ \\

        \midrule

        $\underline{s} \ge 0$ \;\;or\;\; $\overline{s} \ge 0$ &
        $\underline{s}\underline{l}$ &
        $\overline{s}\overline{l}$ \\

        $\underline{s} < 0$ \;\;or\;\; $\overline{s} < 0$ &
        $\underline{s}\overline{l}$ &
        $\overline{s}\underline{l}$ \\

        \bottomrule
    \end{tabular}
\caption{
    Constraint-agnostic bounds on the optimal objective value
    ($\underline{\omega} \le \hat\omega \le \overline{\omega}$)
    that lead to $(1+\epsilon)$-approximations.
    These bounds are defined over existing bounds on the validation scenarios and package size:
    $\underline{s} \le \hat{s}_{ij} \le \overline{s}$, for all $i \in [1..N]$ and $j \in [1..\hat{M}]$,
    and $0 \le \underline{l} \le \sum_{i=1}^{N} \hat{x}_i \le \overline{l}$.
}
\label{tab:generic-bounds}
\end{table}
\endgroup

\begingroup
\renewcommand{\arraystretch}{1.1}
\begin{table*}
\centering
    \begin{tabular}{cc|ll|ll}
        \toprule

        \textbf{Objective-Constraint Interaction} &
        \textbf{Case} &
        $\underline\omega^\odot$ &
        $\underline\omega^\otimes$ &
        $\overline\omega^\odot$ &
        $\overline\omega^\otimes$ \\

        \midrule

        \multirow{2}{*}{(a) Independent} &
        $\underline{s} \ge 0$ \;\;or\;\; $\overline{s} \ge 0$ &
        $p\underline{s}\underline{l}$ &
        $0$ &
        $\overline{s}\overline{l}$ &
        $(1-p)\overline{s}\overline{l}$ \\

        &
        $\underline{s} < 0$ \;\;or\;\; $\overline{s} < 0$ &
        $\underline{s}\overline{l}$ &
        $(1-p)\underline{s}\overline{l}$ &
        $p\overline{s}\underline{l}$ &
        $0$ \\

        \midrule

        \multirow{4}{*}{(b) Supporting/counteracting} &
        $\sum_{i=1}^{N}\xi_i x_i \ge v \ge 0$ &
        $pv$ &
        $0$ &
        --- &
        --- \\

        &
        $\sum_{i=1}^{N}\xi_i x_i \ge v, v < 0$ &
        $v$ &
        $(1-p)v$ &
        --- &
        --- \\

        &
        $\sum_{i=1}^{N}\xi_i x_i \le v, v \ge 0$ &
        --- &
        --- &
        $v$ &
        $(1-p)v$ \\

        &
        $\sum_{i=1}^{N}\xi_i x_i \le v < 0$ &
        --- &
        --- &
        $pv$ &
        $0$ \\

        \bottomrule
    \end{tabular}
\caption{
    Constraint-specific bounds for an objective with inner function $\sum_{i=1}^{N} \xi_i x_i$
    subject to a probabilistic constraint with right-hand side $p$.
    For group (a) in this table, the probabilistic constraint is independent of the objective function;
    for group (b), the constraint supports or counteracts the objective:
    $\probs{\sum_{i=1}^{N} \xi_i x_i \odot v} \ge p$.
    The final bounds on the optimal objective value are of the form
    $\underline\omega^\odot + \underline\omega^\otimes \le \hat\omega \le
    \overline\omega^\odot + \overline\omega^\otimes$.
    An entry with --- means that no bound exists for that case.
    The final bound is the best bound from this table where any of the cases are true.
    For example, for $\underline\omega^\odot$, the final bound is the \emph{maximum} of $p\underline{s}\underline{l}$
    (or $\underline{s}\overline{l}$) and $pv$ (or $v$).
    Notice how there is always at least one available bound from group (a),
    and possibly one additional bound from group (b),
    so that at least one bound for each column of the table always exists.
}
\label{tab:dependent-objective-bounds}
\end{table*}
\endgroup

\subsubsection*{Proofs}
Consider a
counteracting probabilistic constraint
$\probs{\sum_{i=1}^{N} \xi_i x_i \odot v} \ge p$.
We prove that
$\hat\omega \ge \underline\omega^\odot + \underline\omega^\otimes$
under condition ${v \ge 0}$,
which was mentioned in \Cref{subsec:early-termination} for a minimization objective
with a counteracting constraint.
Following the table for case $v \ge 0$, we have to prove $\hat\omega \ge pv$:
\begin{align*}
    \hat\omega
    &= \hat\omega^\odot + \hat\omega^\otimes
    = \frac{1}{\hat M} \sum_{j \in \hat{S}^\odot_{\hat{x}}} \sum_{i=1}^{N} \hat{s}_{ij}.\attr{A}\;\hat{x}_i
    + \frac{1}{\hat M} \sum_{j \in \hat{S}^\otimes_{\hat{x}}} \sum_{i=1}^{N} \hat{s}_{ij}.\attr{A}\;\hat{x}_i \\
    &\ge \frac{1}{\hat M} \sum_{j \in \hat{S}^\odot_{\hat{x}}} v
    = \frac{1}{\hat M} |\hat{S}^\odot_{\hat{x}}| v
    \ge pv,
\end{align*}
where we used the facts that
$|\hat{S}^\otimes_{\hat{x}}| \ge 0$ and
$|\hat{S}^\odot_{\hat{x}}| \ge p\hat{M}$.
We now prove that
$\hat\omega \le \overline\omega^\odot + \overline\omega^\otimes$
under conditions $\overline{s} \ge 0$ and $v \ge 0$,
which was also mentioned in \Cref{subsec:early-termination}
for a minimization objective with supporting constraint.
Following the table, one possible such bound under these conditions
is ${\hat\omega \le v + (1-p)\overline{s}\overline{l}}$, which follows from:
\begin{align*}
    \hat\omega
    &= \hat\omega^\odot + \hat\omega^\otimes
    = \frac{1}{\hat M} \sum_{j \in \hat{S}^\odot_{\hat{x}}} \sum_{i=1}^{N} \hat{s}_{ij}.\attr{A}\;\hat{x}_i
    + \frac{1}{\hat M} \sum_{j \in \hat{S}^\otimes_{\hat{x}}} \sum_{i=1}^{N} \hat{s}_{ij}.\attr{A}\;\hat{x}_i \\
    &\le \frac{1}{\hat M} \sum_{j \in \hat{S}^\odot_{\hat{x}}} v
    + \frac{1}{\hat M} \sum_{j \in \hat{S}^\otimes_{\hat{x}}} \sum_{i=1}^{N} \overline{s}\;\hat{x}_i
    \le \frac{1}{\hat M} |\hat{S}^\odot_{\hat{x}}| v
    + \frac{1}{\hat M} |\hat{S}^\otimes_{\hat{x}}| \overline{s}\overline{l} \\
    &\le v + (1-p)\overline{s}\overline{l},
\end{align*}
where we used the facts that
$|\hat{S}^\odot_{\hat{x}}| \le \hat{M}$ and
$|\hat{S}^\otimes_{\hat{x}}| \le (1-p)\hat{M}$.
All other bounds in the table are easily derivable following a similar approach.

\section{Workload Details} \label{sec:workload}

\begin{figure}[t]
 \footnotesize
 \begin{tabular}{l}
  \textbf{\small Galaxy query template (counteracted objective)} \\
  \SELECT~\PACKAGE{\ast} \FROM~Galaxy \SUCHTHAT \\
  \;\;\;$\COUNT{\ast}$ \BETWEEN $5$ \AND $10$ \AND \\
  \;\;\;$\SUM{\text{Petromag\_r}} \ge \textbf{\{v\}}$ $\WITHPROBABILITY{\ge \textbf{\{p\}}}$ \\
  \MINIMIZE $\EXPECTED{\SUM{\text{Petromag\_r}}}$ \\
  \\
  \textbf{\small Galaxy query template (supported objective)} \\
  \SELECT~\PACKAGE{\ast} \FROM~Galaxy \SUCHTHAT \\
  \;\;\;$\COUNT{\ast}$ \BETWEEN $5$ \AND $10$ \AND \\
  \;\;\;$\SUM{\text{Petromag\_r}} \le \textbf{\{v\}}$ $\WITHPROBABILITY{\ge \textbf{\{p\}}}$ \\
  \MINIMIZE $\EXPECTED{\SUM{\text{Petromag\_r}}}$ \\
  \\
  \textbf{\small Portfolio query template (supported objective)} \\
  \SELECT~\PACKAGE{\ast} \FROM~Stock\_Investments \SUCHTHAT \\
  \;\;\;$\SUM{\text{price}} \le 1000$ \AND \\
  \;\;\;$\SUM{\text{Gain}} \ge \textbf{\{v\}}~\WITHPROBABILITY{\ge \textbf{\{p\}}}$ \\
  \MAXIMIZE $\EXPECTED{\SUM{\text{Gain}}}$ \\
  \\
  \textbf{\small \tpch query template (independent objective)} \\
  \SELECT~\PACKAGE{\ast} \FROM~Tpch\_\textbf{\{D\}} \SUCHTHAT \\
  \;\;\;$\COUNT{\ast}$ \BETWEEN $1$ \AND $10$ \AND \\
  \;\;\;$\SUM{\text{Quantity}} \le \textbf{\{v\}}~\WITHPROBABILITY{\ge \textbf{\{p\}}}$ \\
  \MAXIMIZE $\PROBABILITYOF{\SUM{\text{Revenue}} \ge 1000}$
 \end{tabular}
 \caption{Query templates for the three workloads used in the experimental evaluation of \sss and \naive.
 Each parameter in a template is indicated in curly brackets.}
 \label{fig:example-queries}
\end{figure}

In this section, we provide additional details about the workloads used in our experiments
to evaluate \sss and \naive (\Cref{sec:experiments}).
Figure~\ref{fig:example-queries} shows the \spaql query templates for each dataset.
The parameters in the templates are indicated under curly brackets.
Table~\ref{tab:description} provides all the remaining details for the datasets and queries,
including all the query parameters.

{
\renewcommand{\arraystretch}{0.8}
\small

\newcommand\minE{$\min \expe{\cdot}$}
\newcommand\maxE{$\max \expe{\cdot}$}
\newcommand\maxPr{$\max \prob{\cdot}$}

\newcommand\support{Supported}
\newcommand\counter{Counteracted}
\newcommand\independent{Independent}

\newcommand\highvar{High VaR}
\newcommand\lowvar{Low VaR}
\newcommand\short{2-day}
\newcommand\longpred{1-week}
\newcommand\allstocks{All stocks}
\newcommand\volatile{Most volatile}

\newcommand{\normal}{{\textsc{Normal}}}
\newcommand{\pareto}{\textsc{Pareto}}
\newcommand{\expdis}{$\textsc{Exponential}$}
\newcommand{\pois}{\textsc{Poisson}}
\newcommand{\uni}{\textsc{Uniform}}
\newcommand{\gbm}{\textsc{Geometric} \\ \textsc{Brownian} \\ \textsc{Motion}}
\newcommand{\studt}{\textsc{Student's t}}

\setlength{\tabcolsep}{4pt}

\begin{table*}[t]
\centering
\begin{tabular}{@{}lll|cccllll@{}}
\toprule
\multicolumn{3}{c}{\textbf{Dataset}} & \multicolumn{6}{c}{\textbf{Query}} \\
\midrule
 & \textbf{$N$} & \textbf{Uncertainty} &  & \textbf{Feasible?} & \textbf{Objective} & \textbf{Supportiveness} & \textbf{$p$} & \textbf{$v$} & \textbf{Other features} \\
\midrule
\multirow{8}{*}{\rotatebox{90}{\textbf{Galaxy}}} & \multirow{8}{*}{\begin{tabular}[c]{@{}c@{}}$55,000$ to \\ $274,000$\end{tabular}}
    & \normal($\sigma$=2) & Q1 & \multirow{8}{*}{Yes} & \multirow{8}{*}{\minE} & \multirow{2}{*}{\counter} & \multirow{8}{*}{0.9} & $40$ & \\
 &  & \normal($\sigma^*$=3) & Q2 &  &  &  &  & $43$ & \\
 &  & \normal($\sigma$=2) & Q3 &  &  & \multirow{2}{*}{\support} &  & $50$ & \\
 &  & \normal($\sigma^*$=3) & Q4 &  &  &  &  & $52$ &\\
 &  & \pareto($\sigma$=$\alpha$=1) & Q5 &  &  & \multirow{2}{*}{\counter} &  & $65$ & \\
 &  & \pareto($\sigma^*$=$\alpha$=1) & Q6 &  &  &  &  & $65$ & \\
 &  & \pareto($\sigma$=$\alpha$=1) & Q7 &  &  & \multirow{2}{*}{\support} &  & $109$ & \\
 &  & \pareto($\sigma^*$=3, $\alpha$=1)& Q8 &  &  &  &  & $90$ & \\
\midrule
\multirow{8}{*}{\rotatebox{90}{\textbf{Portfolio}}} & \multirow{8}{*}{\begin{tabular}[c]{@{}c@{}}$4,000$ to \\ $14,000$\end{tabular}} & \multirow{8}{*}{\begin{tabular}[c]{@{}l@{}}\gbm\end{tabular}} & Q1 & \multirow{8}{*}{Yes} & \multirow{8}{*}{\maxE} & \multirow{8}{*}{\support}
                     & 0.9 & $-10$ & \short, \allstocks \\
 &  &  & Q2 &  &  &  & 0.95 & $-10$ & \short, \allstocks \\
 &  &  & Q3 &  &  &  & 0.9 & $-10$ & \short, \volatile \\
 &  &  & Q4 &  &  &  & 0.95 & $-10$ & \short, \volatile \\
 &  &  & Q5 &  &  &  & 0.9 & $-1$ & \short, \volatile \\
 &  &  & Q6 &  &  &  & 0.95 & $-1$ & \short, \volatile \\
 &  &  & Q7 &  &  &  & 0.9 & $-10$ & \longpred, \volatile \\
 &  &  & Q8 &  &  &  & 0.9 & $-1$ & \longpred, \volatile \\
 \midrule
\multirow{8}{*}{\rotatebox{90}{\textbf{TPC-H}}} & \multirow{8}{*}{$117,600$}
    & \expdis($\lambda$=1) & Q1 & \multirow{7}{*}{Yes} & \multirow{8}{*}{\maxPr} & \multirow{8}{*}{\independent} & 0.9 & $15$ & D=3 \\
 &  & \expdis($\lambda$=1) & Q2 &  &  &  & 0.95 & $7$ & D=10 \\
 &  & \pois($\lambda$=2) & Q3 &  &  &  & 0.9 & $15$ & D=3 \\
 &  & \pois($\lambda$=1) & Q4 &  &  &  & 0.9 & $10$ & D=10 \\
 &  & \uni(0,1) & Q5 &  &  &  & 0.9 & $15$ & D=3 \\
 &  & \uni(0,1) & Q6 &  &  &  & 0.95 & $7$ & D=10 \\
 &  & \studt($\nu$=2) & Q7 &  &  &  & 0.9 & $29$ & D=3 \\
 &  & \studt($\nu$=2) & Q8 & No &  &  & 0.95 & $7$ & D=10 \\
 \bottomrule
\end{tabular}
\caption{Detailed description of datasets and queries.
For Galaxy, the means of the distributions are always the original data values,
and we thus only indicate the other distribution parameters (standard deviation $\sigma$ and shape $\alpha$);
The standard deviations can be of two kinds: all identical (indicated by $\sigma$),
or all different and randomly generated (indicated by $\sigma^*$);
In the second case, the standard deviations of the tuples were generated randomly using
a normal distribution with mean zero and standard deviation $\sigma^*$, and by then taking their absolute values.
For Portfolio, ``\short'' indicates predictions made only for the following two days,
and ``\longpred'' indicates predictions for an entire week;
``\volatile'' indicates that the dataset only includes the 30\% most volatile stocks.
For \tpch, we indicate the distribution used to model the data integration uncertainty,
and $D$, the number of integrated sources.}
\label{tab:description}
\end{table*}
}

\end{document}